\documentclass{article}

\usepackage{fullpage,amsmath,amsthm,amssymb}
\usepackage[usenames]{color}

\DeclareMathOperator*{\E}{\mathbb{E}}
\let\Pr\relax
\DeclareMathOperator*{\Pr}{\mathbb{P}}
\DeclareMathOperator*{\Var}{Var}
\DeclareMathOperator*{\poly}{poly}
\DeclareMathOperator*{\Unif}{Unif}
\DeclareMathOperator*{\lsb}{lsb}
\DeclareMathOperator*{\median}{median}
\DeclareMathOperator*{\EQ}{EQ}
\DeclareMathOperator*{\polylog}{polylog}

\usepackage{todonotes}
\usepackage{thmtools}
\usepackage{thm-restate}
\newcommand\Oh{\mathcal{O}}

\renewcommand{\log}{\lg}
\renewcommand{\d}{\,\mathrm{d}}
\newcommand\bR{\mathbb{R}}
\newcommand\bN{\mathbb{N}}
\newcommand\bZ{\mathbb{Z}}
\renewcommand\P{\Pr}

\newcommand{\Gsamp}{\Xi}

\newtheorem{theorem}{Theorem}
\newtheorem{lemma}[theorem]{Lemma}

\newtheorem{corollary}[theorem]{Corollary}
\newtheorem{Definition}[theorem]{Definition}

\newtheorem{Remark}[theorem]{Remark}
\newtheorem{fact}[theorem]{Fact}

\newcommand{\EquationName}[1]{\label{eq:#1}}

\newcommand{\LemmaName}[1]{\label{lem:#1}}

\newcommand{\CorollaryName}[1]{\label{cor:#1}}
\newcommand{\SectionName}[1]{\label{sec:#1}}
\newcommand{\TheoremName}[1]{\label{thm:#1}}
\newcommand{\RemarkName}[1]{\label{rem:#1}}

\newcommand{\FactName}[1]{\label{fac:#1}}

\newcommand{\Equation}[1]{Eq.\:\eqref{eq:#1}}

\newcommand{\LemmaR}[1]{Lemma~\ref{lem:#1}}

\newcommand{\Corollary}[1]{Corollary~\ref{cor:#1}}
\newcommand{\Section}[1]{Section~\ref{sec:#1}}
\newcommand{\TheoremR}[1]{Theorem~\ref{thm:#1}}
\newcommand{\RemarkR}[1]{Remark~\ref{rem:#1}}

\newcommand{\Fact}[1]{Fact~\ref{fac:#1}}

\title{Optimal streaming and tracking distinct elements \\ with high probability}

\author{
	Jaros\l aw B\l asiok\thanks{ John A. Paulson School of Engineering and Applied Sciences, Harvard University, 33 Oxford Street,
Cambridge, MA 02138, USA. Email: {\tt jblasiok@g.harvard.edu.} Supported by ONR grant N00014-15-1-2388.
}}

\begin{document} 

\maketitle

\begin{abstract}
The distinct elements problem is one of the fundamental problems in streaming algorithms --- given a stream of integers in the range $\{1,\ldots,n\}$, we wish to provide a $(1+\varepsilon)$ approximation to the number of distinct elements in the input. After a long line of research an optimal solution for this problem with constant probability of success, using $\mathcal{O}(\frac{1}{\varepsilon^2}+\log n)$ bits of space, was given by Kane, Nelson and Woodruff in 2010.

The standard approach used in order to achieve low failure probability $\delta$ is to take the median of $\log \delta^{-1}$ parallel repetitions of the original algorithm. We show that such a multiplicative space blow-up is unnecessary: we provide an optimal algorithm using $\mathcal{O}(\frac{\log \delta^{-1}}{\varepsilon^2} + \log n)$ bits of space --- matching known lower bounds for this problem. That is, the $\log\delta^{-1}$ factor does not multiply the $\log n$ term. This settles completely the space complexity of the distinct elements problem with respect to all standard parameters.

We consider also the \emph{strong tracking} (or \emph{continuous monitoring}) variant of the distinct elements problem, where we want an algorithm which provides an approximation of the number of distinct elements seen so far, at all times of the stream. We show that this variant can be solved using $\mathcal{O}(\frac{\log \log n + \log \delta^{-1}}{\varepsilon^2} + \log n)$ bits of space, which we show to be optimal. 
\end{abstract}

\section{Introduction}

Estimating the number of distinct elements in the data stream is one of the first, and one of the most fundamental problems in streaming algorithms. In this problem, we observe a \emph{data stream}, i.e. a sequence of elements  $x^{(1)}, x^{(2)}, \ldots, x^{(T)} \in \{1, \ldots n\}$, and we wish to provide a $(1 + \varepsilon)$-approximation for the number of distinct elements in this sequence, using small space $S$. This can be trivially achieved with $\Oh(\min(n, T \log n))$ bits of memory by either storing all elements encountered in the stream, or  by storing a bitmask, keeping a single bit for every possible element of the universe. We wish to provide a probabilistic algorithm using significantly smaller space (allowing for small failure probability).

This problem was first studied by Flajolet and Martin in their seminal paper \cite{DBLP:conf/focs/FlajoletM83} in FOCS 1983, which started a long line of research on subsequently improved algorithms \cite{DBLP:conf/stoc/AlonMS96, DBLP:conf/random/Bar-YossefJKST02, DBLP:conf/spaa/GibbonsT01, Durand2003,DBLP:journals/ton/EstanVF06,flajolet2007hyperloglog, DBLP:conf/vldb/Gibbons01}.  

Kane, Nelson and Woodruff in 2010 \cite{DBLP:conf/pods/KaneNW10} proposed an optimal algorithm for counting the number of distinct elements in the stream with failure probability $\frac{1}{3}$ --- their algorithm provided an $(1+\varepsilon)$ approximation to the a number of distinct elements using $\Oh(\frac{1}{\varepsilon^2} + \log n)$ bits --- the matching lower bound has been shown prior to this \cite{Woodruff:2004:OSL:982792.982817,DBLP:conf/stoc/AlonMS96,DBLP:journals/eccc/BrodyC09}. The standard black-box method of reducing the failure probability of estimation algorithm of this kind is to repeat it independently $\Oh(\log \delta^{-1})$ times in parallel, and use the median of reported answers as the final estimation. This method, applied to the algorithm mentioned above, uses $\Oh(\log \delta^{-1} (\frac{1}{\varepsilon^2} + \log n))$ bits of space.

On the other hand, Jayram and Woodruff in \cite{DBLP:journals/talg/JayramW13} developed a technique for proving lower bounds for streaming problems in the high success probability regime. Their technique allowed them to show that for number of natural streaming problems the naive repetition method is optimal --- for example this is the case for estimation of the $\ell_p$ pseudonorm  (with $p \in [0, 2]$) of frequency vector in the so called \emph{strict turnstile} streaming model. In the same paper they proved a lower bound for the distinct elements problem of form $\Omega(\frac{\log \delta^{-1}}{\varepsilon^2})$. For constant $\varepsilon$, this left a gap between an upper bound $\Oh(\log \delta^{-1} \log n)$ and lower bound $\Omega(\log n + \log \delta^{-1})$. 

It was known that one should \emph{not} expect a lower bound $\Omega(\log \delta^{-1} \log n)$ for this problem. Already \cite{DBLP:conf/pods/KaneNW10} showed that for some constant $\varepsilon$, one can achieve failure probability $\delta = \frac{1}{\poly \log n}$ using only $\Oh(\log n)$ bits, and in \cite{DBLP:journals/talg/JayramW13} it was observed that for every constant $\varepsilon$ there is an algorithm using $\Oh(\log n)$ bits with failure probability $\delta = \frac{\log \log n}{\log n}$. In this paper we completely resolve the question about space complexity of the distinct elements problem in the high success probability regime, showing that the Jayram, Woodruff lower bound was optimal.
\paragraph{Continuous monitoring}
Recently, the space complexity of \emph{tracking} problems in data streams has been considered --- namely we say that streaming algorithm provides \emph{strong tracking} of a statistic $f$ of the input stream, if  after every update it reports quantity $\tilde{f}$ such that
\begin{equation*}
    \P(\forall t, (1-\varepsilon) f^{(t)} \leq \tilde{f}^{(t)} \leq (1+\varepsilon) f^{(t)}) > 1 - \delta.
\end{equation*}

The first result of this form that we are aware of, was proven in \cite{DBLP:conf/pods/KaneNW10} as a subroutine for non-tracking estimation of distinct elements. They showed that one can achieve tracking of $F_0$ with some constant approximation factor, using $\Oh(\log n)$ bits of space. The question whether one can achieve strong tracking without the naive union bound over all positions of the stream was explicitly asked later in \cite{HuangTY14}, where they also proposed an algorithm for estimation of the $\ell_p$-pseudonorm of the frequency vector, for $p \in (0,2]$. Their algorithm yields improvement over the baseline approach for very long input streams $\log T = \omega(\log n)$. The strong tracking of $\ell_2$ has been later improved in \cite{BravermanCWY16, BravermanCINWW17}, where interesting results are achieved even in the more standard regime of parameters, with $n$ and $T$ that are polynomially related. They showed that one can solve strong tracking of $\ell_2$ using $\Oh(\frac{\log n \log \log T}{\varepsilon^2})$ bits, as compared to naive bound of form $\Oh(\frac{\log n \log T}{\varepsilon^2})$. The improved algorithm for strong tracking of $\ell_p$ with $0 < p < 2$ was provided in \cite{DBLP:journals/corr/DingBN17}.

\paragraph{Our contribution.}
We provide an optimal streaming algorithm for the distinct elements problem in the high probability regime, using $\Oh(\log n + \frac{\log \delta^{-1}}{\varepsilon^2})$ bits of space. This result completely settles the space complexity of this problem with respect to all standard parameters.

We also show a strong tracking algorithm for the distinct elements with space $\Oh(\log n + \frac{\log \delta^{-1} + \log \log n}{\varepsilon^2})$, together with a matching lower bound --- we prove that $\Omega(\frac{\log \log n}{\varepsilon^2})$ term is necessary. The $\Omega(\frac{\log \delta^{-1}}{\varepsilon^2} + \log n)$ lower bound was already known even for the easier non-tracking version of the distinct elements problem.

This is a first matching lower bound for any strong tracking problem, where the non-trivial algorithm is achievable. This shows a separation between the traditional estimation problem and strong tracking variation when $\varepsilon = o(\sqrt{\frac{\log \log n}{\log n}})$. On the other hand, in the regime $\varepsilon = \Omega(\sqrt{\frac{\log \log n}{\log n}})$ the strong tracking problem is not harder than one-shot estimation (up to constant factors).

The update time of our algorithm is $\poly(\log n, \log \delta^{-1})$. The only bottleneck is the pseudorandom construction described in \Section{sampler-lemma}. In particular, by substituting this construction with a random walk over an expander graph of super-constant degree, it is possible to achieve update time $\Oh(\log \delta^{-1} + \log \log n)$, with slightly worse space complexity $\Oh(\frac{\log \delta^{-1}}{\varepsilon^2} + \log\delta^{-1} \log \log n + \log n)$.

\section{Notation}
For a natural number $m$, by $[m]$ we denote set $\{1,\ldots m\}$.
For a finite set $A$, by $\# A \in \bZ$ or $|A|$ we denote the cardinality of $A$.
For $X \in [2^n]$, we will write $\bar{X} \in \{0, 1\}^n$ to be a bit representation of $X$. For a bitvector $y \in \{0, 1\}^n$ we denote $|y| = \#\{i : y_i = 1\}$. For two bitvectors $x, y$, we take $x \lor y \in \{0,1\}^n$ to be the bitvector with $(x \lor y)_j = 1$ if and only if $x_i = 1$ or $y_i = 1$.

In the paper $n$ will be used to denote the size of the universe from which the elements in the input stream are chosen, $T$ --- the length of the stream, and $x^{(1)}, x^{(2)}, \ldots x^{(t)} \in [n]$ are those elements. Set $S^{(t)} := \{ x^{(1)}, \ldots x^{(t)}\}$ to be the set of all distinct elements seen up to a time step $t$, and $F_0^{(t)} := \# S^{(t)}$. 

Throughout the paper we use notation $A \lesssim B$, to denote the existence of an absolute constant $C$ such that $A \leq C B$, where $A$ and $B$ themselves may depend on a number of parameters. We write $ A \simeq_{1 + \varepsilon} B$ to denote $(1 - \varepsilon)B\leq A \leq (1+\varepsilon) B$. 

\section{Overview of our approach}
\subsection{Constant factor approximation with high probability}
The main goal of \Section{constant-approx} is to show a streaming algorithm that provides an $\Oh(1)$-approximation to the number of distinct elements at all times in the stream (i.e. $\Oh(1)$-strong tracking), with probability $1 - \delta$ using optimal space $\Oh(\log n + \log \delta^{-1})$ bits. That is, we want to provide estimate $\tilde{F}_0(t)$, such that
\begin{equation*}
    \P\left(\forall t,\ \frac{1}{C} F_0^{(t)} \leq \tilde{F}_0^{(t)} \leq C F_0^{(t)}\right) > 1 - \delta
\end{equation*}
where $F_0^{(t)}$ is a number of distinct elements on the input among $x^{(1)}, \ldots x^{(t)}$.

Note that in this regime of parameters, if one has an algorithm estimating number of distinct elements using space complexity $\Oh(\log n + \log \delta^{-1})$, one can set $\delta = \delta'/n$, and apply a union bound over all insertions to the stream, to get a strong tracking algorithm for the same problem with failure probability $\delta'$ and space complexity $\Oh(\log n + \log \delta^{-1})$. As such, we can without loss of generality focus on the strong tracking version, and this stronger guarantee is going to be useful in order to ensure that the algorithm can be implemented using small space.

To discuss the main idea behind our approach, for the sake of presentation we will first consider a random oracle model --- here we assume that the algorithm is augmented with the access to a uniformly random function $R : \{0,1\}^* \to \{0, 1\}$ (all the values of $R$ are uniform and independent); in particular the space to store such a function does not count towards the space complexity of the algorithm, and the failure probability is understood over a selection of the oracle. For space complexity of such an algorithm, we will count only the amount of information passed between observations of elements from the input stream; we are allowed to use larger space to process an element from the input. This allows us to talk in a meaningful way about space complexity $o(\log n)$, even though any single element in the stream already take $\Theta(\log n)$ bits to store.

Let us start with discussion on how to design an algorithm using $\Oh(\log \log n + \log \delta^{-1})$ bits of space in the random oracle 
model. It is well-known that given access to a random hash function $h : [n] \to \{0, 1\}^{\log n}$, if we fix some set $S \subset [n]$, then $\hat{X} := \max_{s \in S} \lsb(h(s))$ is such that $2^{\hat{X}}$ is a constant factor approximation to $|S|$ with probability $2/3$, where $\lsb$ is the least-significant-bit function \cite{DBLP:conf/focs/FlajoletM83}. Indeed, to argue that this is true, we can consider subsets $S_k \subset S$ given by $S_k := \{ s \in S : \lsb(h(s)) \geq k \}$ --- every such subset corresponds to sub-sampling $S$ by a factor $2^k$, and we should expect that the last non-empty set $S_k$ is the one corresponding to sub-sampling by a factor roughly $\frac{1}{|S|} = 2^{-\log |S|}$.

We can repeat an estimator $\hat{X}$ constructed above $\Oh(\log \delta^{-1})$ times independently in parallel and take median, in order to achieve achieve $\Oh(\log \delta^{-1} \log \log n)$ bit complexity for failure probability $\delta$ under the random oracle model. To improve this construction to $\Oh(\log \log n + \log \delta^{-1})$ bits, let us take $\Oh(\log \log n + \log \delta^{-1})$ independent estimators as above. Instead of storing all those estimators $\hat{X}_k$ independently, we can store the median (which takes $\Oh(\log \log n)$ bits), and deviations $\hat{X}_i - \median(\{\hat{X}_1, \ldots \hat{X}_w\})$. One can show that with high probability at all times throughout the stream the median is a good estimator of the number of distinct elements seen so far, and moreover --- because the deviations $\hat{X}_i - \log |S|$ are random variables that are extremely well concentrated around zero --- on average over all the counters we will use constant number of bits per counter to store all deviations from median, at all time steps.

Getting rid of the random oracle assumption is much more technical --- without access to the random oracle, it is known (\cite{DBLP:conf/stoc/AlonMS96}) that one can use a pairwise independent hash function $h : [n] \to \{0, 1\}^{\log n}$ to get a constant success probability --- and a seed to such a hash function can be stored using $\Oh(\log n)$ bits. This, together with median over parallel repetitions of the estimator, yields simple $\Oh(\log n \log \delta^{-1})$ space algorithm with failure probability $\delta$. 

To improve upon that, we can observe that in this setting it is not necessary for all the $\Oh(\log \log n + \log \delta^{-1})$ estimators to use independent seeds for the underlying pairwise-independent hash functions $h_i$. Instead, we can consider a fully explicit constant degree expander graph, with the set of vertices $[N]$ corresponding to the set of seeds for pairwise independent hash functions. We would choose the first seed for $h_1$ uniformly at random, but subsequent seeds are chosen by a random walk over this expander graph. In such a way, we can succinctly store all the seeds using $\Oh(\log \delta^{-1} + \log n)$-bits of space, and the standard Chernoff-bounds for expander walks \cite[Theorem 4.22]{Vadhan12} imply that median of estimators generated in such a way is still constant factor approximation for the number of distinct elements, except with small failure probability $\delta$. This yields an algorithm with space complexity $\Oh(\log n + \log \log n \log \delta^{-1})$, if we store $\hat{X}_i$ explicitly --- still falling short of our goal of $\Oh(\log n + \log \delta^{-1})$ bits of space.

Unfortunately, we \emph{cannot} argue, as before in the random oracle model, that we can succinctly store all counters $\hat{X}_i$ generated via such an expander walk by considering the median and deviations from the median separately --- sufficiently strong concentration bounds are \emph{not true} for a constant degree expander walk.

Instead, inspired by the construction of a sampler in \cite{DBLP:conf/soda/Meka17}, we show that by composing a number of pseudorandom objects (i.e. pairwise independent hash functions, short walks over super-constant expander graphs, averaging samplers obtained from the celebrated construction of strong extractors \cite{DBLP:journals/rsa/Zuckerman97}, and standard sub-sampling methods), we can generate $w = \Theta(\log n + \log \delta^{-1})$ estimators in total, divided into groups of estimators.  More concretely, we produce $\frac{w}{\log w}$ groups of estimators, such that each group has about $\log w$ estimators, and with the probability $1 - \exp(-\Omega(w)) = 1 - \delta$ for at least half of the groups the median yields a good estimation of $F_0$ at all times, while simultaneously the ``good'' groups take at all times $\Oh(1)$-bits on average per estimator to store, if we store estimators within a group by storing separately the median and  deviations from the median, as discussed above. 

It is essential for this argument that size of each group $w_1 = \Oh(\log w)$ is greater than $C \log \log n$ --- intuitively, if we consider a random group of such a size, the  probability that we need too many bits to store compactly such a group at any fixed step $t$ is bounded by $\exp(-\Omega(w_1)) = \frac{1}{\polylog(n)}$, and therefore we can union bound over all $\Oh(\log n)$ positions where $F_0$ grows by factor of two, without affecting the failure probability too much.

The details of this pseudorandom construction are presented in \Section{sampler-lemma}. This is the main technical difficulty in proving the following theorem.

\begin{theorem}
	\TheoremName{constant-approx}
    There is a streaming algorithm with space complexity $\Oh(\log n + \log \delta^{-1})$ bits, that with probability $1-\delta$ reports a constant factor approximation to number of distinct elements in the stream after every update.
\end{theorem}

This space complexity is optimal \cite{DBLP:conf/stoc/AlonMS96,DBLP:journals/talg/JayramW13}. 

The algorithm could be significantly simplified, and would mimic exactly the algorithm in the random oracle scenario, if we had an explicit sampler satisfying the following guarantee, with seed length $s = \Oh(w + \log n)$. In the following definition $\text{Unif}(S)$ denotes the uniform probability distribution over a set $S$.

\begin{Definition}[Sub-gaussian sampler]
		Function $\Gamma : \{0, 1\}^s \times [w] \to [n]$ is called a sub-gaussian sampler if and only if for every $f : [n] \to \bR$ satisfying $\P_{X\sim \text{Unif}([n])}(|f(X)| > \lambda) < \exp(-\lambda^2)$, we have
		\begin{equation*}
				\P_{S \sim \Unif(\{0,1\}^s)}(\sum_{i \leq w} f(\Gamma(S, i)) > C w) < \exp(-\Omega(w)).
		\end{equation*}
\end{Definition}

We consider existence of explicit samplers like this with seed length $\Oh(w + \log n)$ to be an interesting question on its own in the area of pseudorandomness, that will likely have many other applications
\footnote{Non-constructive existence of samplers like that can be proven using probabilistic method, and reduction to $\gamma$-Strong samplers similar to \LemmaR{strong-implies-unbounded-tails}.}. In fact for black-box derandomization of the random oracle algorithm described in this section it is enough to have sampler for functions with stronger tail probabilities --- it is enough for it to apply to functions with doubly exponential tails $\P(|f(X)| > \lambda) < \exp(-e^\lambda)$.

\begin{Remark}
    The update time of this algorithm is $\Oh(\poly(\log \delta^{-1}, \log n))$. The only bottleneck is the pseudorandom construction we are using. If we give up on succinctly storing estimates $\hat{X}_i^{(t)}$, and store them explicitly, we can replace this pseudorandom construction with a single random walk over constant degree expander graph. There are expander graphs that allow evaluation of the neighbour function in constant time \cite{gabber1981explicit}. Such a modification would give an algorithm using slightly worse space $\Oh(\log n + \log \delta^{-1} \log \log n)$, but $\Oh(\log \delta^{-1} + \log \log n)$ update time for strong tracking $F_0$ with constant factor approximation.

    It is possible to carry this construction over to our subsequent result, achieving $\Oh(\log \delta^{-1} (\frac{1}{\varepsilon^2} + \log \log n) + \log n)$ bits of space for high accuracy regime, and $\Oh(\log \delta^{-1} \frac{\log \log n}{\varepsilon^2} + \log n)$ bits of space for tracking, all with update time $\Oh(\log \log n + \log \delta^{-1})$.
\end{Remark}

\subsection{High accuracy regime}
In \Section{high-accuracy} we discuss how to use the previous construction to achieve a high accuracy estimation of the number of distinct elements, with probability $1-\delta$. We prove the following theorem. 
\begin{theorem}
    \TheoremName{high-accuracy}
    For every $\varepsilon, \delta$ there is an algorithm using $\Oh(\frac{\log \delta^{-1}}{\varepsilon^2} + \log n)$ bits, which, at the end of the stream reports $1 + \varepsilon$ approximation to the number of distinct elements, with probability $1-\delta$.
\end{theorem}
\begin{Remark}{\cite{DBLP:conf/stoc/AlonMS96,DBLP:journals/talg/JayramW13}}
    This space complexity is optimal --- every algorithm that estimates number of distinct elements up to a $(1+\varepsilon)$ factor with probability $1-\delta$ needs to use space at least $\Omega(\frac{\log\delta^{-1}}{\varepsilon^2} + \log n)$ bits.
\end{Remark}


Given ideas in the work of Kane \emph{et al.} \cite{DBLP:conf/pods/KaneNW10} and results obtained in the previous section, getting correct dependence on the error parameter $\varepsilon$ is routine, although somewhat tedious.

We consider separately two ranges of parameters: if $\varepsilon < \left(\frac{1}{\log n}\right)^{1/4}$, the KNW algorithm (given access as a black-box to the strong tracking with some constant approximation), using space $\Oh(\frac{1}{\varepsilon^2})$ and $\Oh(\log n)$ random bits has probability $\frac{2}{3}$ of providing $(1+\varepsilon)$ approximation to the number of distinct elements --- since $\varepsilon$ is small here, the space budget $\Oh(\frac{1}{\varepsilon^2})$ is large enough for the analysis to work. We can instantiate $\Oh(\log \delta^{-1})$ parallel copies of this algorithm, providing them access to the same strong-tracker with failure probability $\delta/2$. Naively, we would have to store $\Oh(\log n \log \delta^{-1})$ random bits in order to do this, each instance of KNW algorithm is using $\Oh(\log n)$ random bits --- to reduce the amount of randomness necessary, we pick them using a walk over a constant-degree expander graph. That is, random bits for first instance of a KNW algorithm are completely uniform, but bits for subsequent runs of KNW are chosen by following a random edge of an expander graph. We can use standard Chernoff-bounds for expander walks, as in \cite{Gillman:1998:CBR:284943.284979}, to show that failure probability of such an algorithm is still at most $\delta$.

On the other hand, if $\varepsilon > \left(\frac{1}{\log n}\right)^{1/4}$, we can assume without loss of generality that $\log \delta^{-1} > \Omega(\sqrt{\log n})$, because $\frac{\sqrt{\log n}}{\varepsilon^2} = o(\log n)$ anyway, and our target is space complexity of form $\Oh(\frac{\log \delta^{-1}}{\varepsilon^2} + \log n)$. In this case, we can instantiate $\Theta(\sqrt{\log n} + \log \delta^{-1})$ parallel copies of the KNW algorithm, using the pseudorandom construction as described in \Section{sampler-lemma}. Here the number of instantiations of this algorithm is large enough, therefore we hope that the space consumption at every time step, on average over all the instatiations will be small. Identical analysis as for a constant approximation factor in \Section{constant-approx} can be used to deduce correctness of such an approach.

In fact the space guarantees of the KNW algorithm, as it was originally analyzed, applied only when $\varepsilon  \leq \sqrt{\log n}$ --- as this could be assumed without loss of generality in the original setting. We provide a more delicate analysis of the space consumption of this algorithm in Appendix~\ref{sec:appendix} (specifically Theorem~\ref{thm:knw-reduction}), that is sufficient for our purposes.

\subsection{Strong tracking of distinct elements}
In \Section{strong-tracking} we discuss how to achieve the $(1+\varepsilon)$-strong tracking guarantee for $F_0$ estimation. 

First, let us observe that an algorithm estimating $F_0$ with small failure probability already translates into some upper bound on the space complexity for the tracking problem. Given that number of distinct elements in the stream is increasing, and our estimators proposed in \Section{high-accuracy} are monotone as well, it is enough to look at a sequence of positions $t_1, t_2, \ldots t_s$ such that $F_0^{(t_i)} = (1 + \varepsilon)^{i}$. If the estimate is within $(1+\varepsilon)$ from the actual number of distinct elements at all points $t_i$, we can deduce a strong tracking with accuracy $1 + \Oh(\varepsilon)$: for $t_i \leq t \leq t_{i+1}$ we have $\hat{F}_0^{(t)} \leq \hat{F}_{0}^{(t_{i+1})} \leq (1 + \varepsilon) F_0^{(t_{i+1})} \leq (1 + \varepsilon)^2 F_0^{(t)}$, and similarly for the lower bound. There are at most $\log_{1+\varepsilon} n = \Oh(\frac{\log n}{\varepsilon})$ such positions $t_i$, so by setting failure probability $\delta := \delta'\left(\frac{\log n}{\varepsilon}\right)^{-1}$ in \TheoremR{high-accuracy}, we can deduce that there is an algorithm satisfying $1+\varepsilon$ strong tracking of $F_0$ with probability $1-\delta$, using $\Oh(\frac{\log \log n + \log \delta^{-1}}{\varepsilon^2} + \log n + \frac{\log \varepsilon^{-1}}{\varepsilon^2})$ bits of space.

We show that, by opening up the \cite{DBLP:conf/pods/KaneNW10} construction and more detailed analysis, it is possible to remove the additive $\frac{\log \varepsilon^{-1}}{\varepsilon^2}$ term, and obtain an optimal algorithm for $F_0$ tracking.

\begin{theorem} \TheoremName{strong-tracking}
    There is an algorithm for $1+\varepsilon$ strong tracking of the number of distinct elements in the stream, using $\Oh(\frac{\log \log n + \log \delta^{-1}}{\varepsilon^2} + \log n)$ bits of space.
\end{theorem}

To describe the overview of our contribution, let us first discuss the high-level idea behind the \cite{DBLP:conf/pods/KaneNW10} algorithm. Let us focus first on the random oracle model. Consider a fixed set $S \subset [n]$ (the set of distinct elements seen at the end of the stream), a random hash function $h : [n] \to \{0, 1\}^{\log n}$, and sets $S_k := \{s \in S : \lsb(s) \geq k\}$ --- those sets correspond roughly to subsampling $S$ by a factor of $2^k$. If we already have access to a constant factor approximation of $|S|$, we can zoom in onto set $S_k$ for which we expect $|S_k| = \Theta(\frac{1}{\varepsilon^2})$. Clearly $2^k |S_k|$ is an unbiased estimator of $|S|$, and moreover the standard deviation of $2^k |S_k|$ is of order $\Oh(\varepsilon |S|)$. This implies that if we had a way to estimate size of $|S_k|$ up multiplicative factor $(1+\varepsilon)$ that would be enough to get an $(1 + \Oh(\varepsilon))$ approximation for $|S|$.

In order to do this, we can check a hash function $h_2: [n] \to [P]$ for $P \approx \frac{100}{\varepsilon^2}$. We wish to recover $|S_k|$ from $|h_2(S_k)|$. This is reminiscent of a famous balls-and-bins thought experiment: we are throwing $|S_k|$ balls randomly into $P$ bins, and we try to estimate number of balls, given number of non-empty bins. Let us define $\Phi_P(t)$ to be the expected number of non-empty bins, after throwing $t$ balls at random into $P$ bins (we will drop the subscript $P$ in further discussion), we have $\E_{h_2} |h_2(S_k)| = \Phi_P(|S_k|)$. We claim that, as long as $|S_k| \leq \frac{P}{20}$, we will have with good probability $|S_k| \approx_{1 + \varepsilon} \Phi^{-1}(|h_2(S_k)|)$.  This is because $\sqrt{\Var(|h_2(S_k)|)} = \Theta(\sqrt{|S_k|}) = \Theta(\varepsilon |S_k|)$, so $|h_2(S_k)| = \Phi(|S_k|) \pm \Oh(\varepsilon |S_k|)$, but $\Phi_P$ is bi-Lipschitz in the regime $t \leq P/20$  --- i.e. for any $a \leq b \leq P/20$ we have $0.9(b - a) \leq \Phi(b - a) \leq b - a$. We can put those two facts together to get $|S| \approx_{1+\varepsilon} 2^k |S_k| \approx_{1+\varepsilon} 2^k \Phi^{-1}(|h_2(S_k)|)$.

Using $\Oh(\frac{\log \log n}{\varepsilon^2})$ bits in total, we can have access to $|h_2(S_k)|$ for all $k$ throughout the stream --- it is enough for each $p \in [P]$ to store $\max \{ \lsb(s) : s \in S, h_2(s) = p\}$. In \cite{DBLP:conf/pods/KaneNW10} it is discussed, among other things, how to reduce the space complexity of storing the information about $h_2(S_k)$ for all $k$ to $\Oh(1/\varepsilon^2)$-bits in total, and how to remove the random oracle assumption, by using compositions of bounded-wise independent hash functions. We describe this algorithm in Appendix~\ref{sec:appendix}, together with more detailed analysis of the distribution of space complexity of this algorithm.

In order to achieve smaller space of the tracking algorithm, let us focus on a specific $k$ and consider evolution of $|S_k^{(t)}|$ over the updates to the stream, where $S^{(t)} := \{ x^{(1)}, \ldots x^{(t)} \}$, and $S_k^{(t)} := \{x \in S^{(t)} : \lsb(s) \geq k\}$. More specifically, let us take $K := 2^k \varepsilon^{-2}$, and let us look at the stream, given the promise that $|S^{(T)}| < K/100$. We wish to say that with probability $2/3$ simultaneously all times $|S_k^{(t)}| 2^k$ gives us an approximation of $|S^{(t)}|$ up to additive term $\pm \varepsilon K$. Moreover, we want to say that $\Phi^{-1}(|h_2(S_k^{(t)})|)$ yields at all time approximation to $|S_k^{(t)}|$, again with additive error $\pm \varepsilon K 2^{-k}$. If we are able to show this, we can later amplify this success probability to $1 - \Theta(\frac{1}{\log n}))$ using $\Oh(\log \log n)$ repetitions of the whole algorithm, and union bound over all possible settings of $k$ in order to achieve strong tracking. Note that there are only $\Oh(\log n)$ values of $k$ to union bound over, as opposed to $\Oh(\varepsilon^{-1} \log n)$ distinct positions in the stream where $F_0^{(t)}$ grows by a multiplicative factor of $(1+\varepsilon)$. 

In the random oracle model both facts --- the fact that for all $t$ we have $|S_k^{(t)}| 2^k \simeq |S^{(t)}| \pm \varepsilon K$, as well as the fact that for all $t$, we have $\Phi^{-1}(|h_2(S_k^{(t)})|) = |S_k^{(t)}| \pm 2^{-k} \varepsilon K$ can be proven by the Doob's martingale inequality. In particular, the fact that $2^k |S_k^{(t)}|$ is an approximation to $|S^{(t)}|$ at all times $t$, follows directly (after shifting and rescaling) from the fact that for a random walk $Y_t := \sum_{i \leq t} X_i$, where $X_i$ are arbitrary random variables satisfying $\E X_i = 0$ and $\E X_i^2 = 1$, we have $\sup_{t \leq T} |Y_t| = \Oh(\sqrt{T})$ with good probability. The main technical difficulty in the strong tracking part of this paper lies in dropping the random oracle assumption, and showing some variation on Doob's martingale inequality under bounded independent hash functions. In particular, we show the following lemma about the deviations of random walk that might be of independent interest
\begin{lemma}
    \LemmaName{4-wise-indep-random-walk}
    Let $X_1, \ldots X_T$ be collection of $4$-wise independent random variables, with $\E X_i = 0$, and $\E X_i^2 = 1$, and let $S_t := \sum_{i \leq t} X_i$, then
    \begin{equation*}
        \P\left(\sup_{t \leq T} |S_t| > \lambda\right) \lesssim \frac{T}{\lambda^2}.
    \end{equation*}
\end{lemma}

A result of the same spirit can be deduced from \cite[Theorem 10]{BravermanCINWW17} when $X_i$ are uniform $\pm 1$ random variables --- in our case, however, the steps $X_i$ are significantly less well-behaved, i.e. $\E X_i^4$ is already extremely large, even compared to $T$.

 \LemmaR{4-wise-indep-random-walk} already implies the first part of the argument: that for all $t$ we have $2^k |S_k^{(t)}| = |S^{(t)}| \pm \Oh(\varepsilon K)$. To control deviations of $h_2(|S_k^{(t)}|)$ from its expectation, we use $h_2$ to be a composition of pairwise independent hash function $h_3 : [n] \to [P^2]$, and $\poly(\log \frac{1}{\varepsilon})$-wise independent function $h_4 : [P^2] \to [P]$. We should expect that $|h_3(S_k^{(T)})| = |S_k^{(T)}|$, i.e. function $h_3$ has no collisions with probability $9/10$, and all we care about are deviations of $|h_4(\tilde{S}^{(t)})|$ from its expectation, where $\tilde{S}^{(t)} \subset [P^2]$ is such that $|\tilde{S}^{(t)}| = |S^{(t)}|$.

 Consider $\phi : [P]^{P/20} \to \bN$, such that $\phi(r_1, r_2, \ldots r_{P/20}) := \# \{ j : \exists i,\ r_i = j\}$ --- in the random oracle model, bounding the deviations $|h_4(\tilde{S}^{(t)})| - \E |h_4(\tilde{S}^{(t)})|$ can be reduced to bounding the deviations of the Doob's martingale $Y_t := \E_{X'} \phi(r_1, \ldots r_t, X'_{t+1}, \ldots X'_T)$, where $X' \sim \Unif([P])$. In this setting the Doob's martingale inequality yields
\begin{equation}
	\P_{r_1, \ldots r_{P/20} \sim \Unif([P])}\left(\sup_{t \leq P/20} |Y_t - \E Y_t| > \lambda\right) \leq \frac{\Var{\phi}}{\lambda^2} 
    \EquationName{doobs-inequality}
\end{equation}
where $x_i$ and $\hat{X}$ are independent and uniform. Finally, this together with bi-Lipschitz property of function $\Phi_P$ in the range of interest, implies that indeed we have $\forall t, \Phi^{-1}(|h_4(h_3(S_k^{(t)}))|) = |S_k^{(t)}| \pm \Oh(\varepsilon K)$.

In our case, variables $r_i$ have only bounded-wise independence, and the process $Y_t$ above is no longer a martingale. We deal with this, by showing that $\phi$ can be approximated (in some sense, under the distributions of interest) by a $\poly(\log P)$-degree polynomial, and we show that under some additional restrictions, processes as above induced by degree $d$ polynomials, satisfy the same \Equation{doobs-inequality}, even if variables $r_i$ are only $4d$-wise independent, as opposed to fully independent.

\subsection{Strong tracking lower bound}
In \Section{lower-bound} we show the optimality of the strong tracking algorithm proposed in the previous section. We prove the following theorem.
\begin{theorem}
    Every algorithm solving $1+\varepsilon$ strong tracking for $F_0$ estimation with probability $\frac{2}{3}$ needs to use $\Omega(\frac{\log \log n}{\varepsilon^2})$ bits of space.
\end{theorem}

Together with previously known lower bound $\Omega(\frac{\log \delta^{-1}}{\varepsilon^2} + \log n)$ for $F_0$ estimation, this shows a lower bound that exactly matches our upper bound discussed earlier.

In order to show this, we introduce a $k$-round communication game, where at round $k$, Alice observes input $x_k \in \{0,1\}^n$, Bob observes input $y_k \in \{0,1\}^n$, and they all observe all the previous inputs $(x_i, y_i)_{i \leq k-1}$. In the $k$-th round, Alice sends a message to Bob, and Bob is supposed to report $(1+ \varepsilon)$ approximation to number of ones in a string $x_k \lor y_k$. The protocol is successful if and only if simultaneously at all rounds Bob reports correct (approximate) answer. We show that existence of a strong tracking algorithm implies low-communication protocol for this kind of game with $k = \Theta(\log n)$ rounds --- which in turn, implies a one-round one-way communication protocol for estimation of $x \lor y$ with small failure probability $\delta = \Theta(1/k)$. This would contradict known communication complexity lower bound for small failure probability of distinct element counting \cite{DBLP:journals/talg/JayramW13}.

\subsection{Pseudorandom construction}

In \Section{sampler-lemma}, we prove the main derandomization lemma used in the algorithm described in \Section{constant-approx} for constant factor approximation of the number of distinct elements. Take $w = \Theta(\log \delta^{-1} + \log \log n)$ --- we wish to use $w$ instantiations of the basic estimator, each instantiation is uniquely determined by a seed to pairwise independent hash function used for the estimator (such a seed is of length $\Oh(\log n)$, let us call the number of different seeds $N$).

We pick $w_2 := \Theta(\log w)$, a random walk of length $w_2$ over an expander graph with vertices $[N]$ and degree $\text{quasipoly}(w_2)$ will be \emph{bad} with probability $\exp(-\Omega(w_2))$ --- by which we mean that either the median of all the estimators produced by this walk is at some point far from actual $F_0$, or that we need at some point more than $C w_2$ bits to store all the values of the estimators by storing median and deviations from median. A single random walk like this is going to need $\Oh(\log n + w_2 \log^2 w_2)$ random bits. If we consider now space of all those random walks, we can use known construction of averaging samplers to get a sample of size $W = \poly(w)$, such that with probability $1 - \exp(-\Omega(w))$ the fraction of failed random walks is the same as in the entire space. If we condition on the event that the sampler succeeded, by taking $\frac{w}{\log w}$ independent elements from the sample $[W]$, we can see that more than half of them is bad with probability $\exp(-c w_1)^{w/\log w} = \exp(-c w)$. As we are taking $\frac{w}{\log w}$ independent elements from the universe of size $\Oh(\poly(w))$ we need only $\frac{w}{\log w} \Oh(\log w) = \Oh(w)$ random bits to achieve this.

\section{Constant factor approximation with high probability \SectionName{constant-approx}}

In this section we prove \TheoremR{constant-approx}, assuming existence of specific pseudo-random objects described in \LemmaR{sampler-lemma}. The proof of \LemmaR{sampler-lemma} itself is postponed until later in \Section{sampler-lemma}. We will first state few necessary definitions, followed by a statement of a pseudo-random lemma, then we proceed with the proof of \TheoremR{constant-approx}.

\begin{Definition}[doubly-exponential tail]
We say that a function $f : [M] \to \mathbb{R}_+$ satisfies doubly-exponential tail bounds if $\P_{s \sim \Unif([M])} (f(s) > \lambda) < \exp(-e^\lambda)$.
\end{Definition}
\begin{Definition}[$C$-small set]
	Consider a finite universe $[M]$, equipped with $R$ functions $g_1, \ldots g_R: [M] \to \mathbb{R}_+$. We will say that a sequence $S \in [M]^*$ is $C$-small with respect to $g_1, \ldots g_R$, if $\forall t\leq R,\ \sum_{X \in S} g_t(X) \leq C |S|$.
\end{Definition}
Equipped with those definitions, we are ready to state the pseudorandom lemma.
\begin{lemma}
    \LemmaName{sampler-lemma}
	For any $w \geq \Omega(R^{1/2})$ satisfying in addition $w \geq (\log M)^{\Omega(1)}$, there exist $w_1, w_2$ with $w_1 w_2 = \Theta(w)$, and an explicit function $\Gsamp : \{0, 1\}^s \times [w_1] \to [M]^{w_2}$, such that for any $g_1, \ldots g_R : [M] \to \bR_+$ with doubly-exponential tail bounds, we have with probability at least $1 - \exp(\Omega(-w))$ over a random selection of the seed $U \in \{0, 1\}^s$, that majority of sequences $\Gsamp(U, k)$ is $C$-small for some universal constant $C$.

    That is
    \begin{equation*}
	    \exists \Gsamp, \forall g_1, \ldots g_R, \P_{U \sim \Unif(\{0,1\}^s)}( \#\{ k : \Gsamp(U, k) \text{ is $C$-small}\} > w_1/2) \geq 1 - \exp(-\Omega(w)),
    \end{equation*}
    where $g_1, \ldots g_R$ above must satisfy doubly-exponential tail bounds.

    The seed length in this construction is $s = \Oh(\log M + w)$, and $\Gsamp$ can be evaluated in time  $\poly(s)$.
\end{lemma}

Let us now proceed with the proof of \TheoremR{constant-approx}.

Fix a stream of updates $x^{(1)}, x^{(t)}, \ldots x^{(T)} \in [n]$, and corresponding sets $S^{(t)} := \{ x^{(i)} : i \leq t\}$.

Consider $[N]$ as a set, with implicit bijection to a family of pairwise independent hash functions from $[n]$ to $[n]$, where $\log N = \Theta(\log n)$. For each $i \in [N]$, we have corresponding hash function $h_i : [n] \to [n]$, and estimates $Y_i^{(t)} := \max \{ \lsb(h_i(s)) : s \in S^{(t)}\}$ --- the estimate for $\lg |S^{(t)}|$ given by hash function $h_i$. We will focus on the error of those estimators $\hat{Y}_i^{(t)} := Y_i^{(t)} - \lg |S^{(t)}|$.

\begin{fact}
	\FactName{subexp-tails}
    The error terms satisfy following subexponential tail bound
	\begin{equation*}
			\P_{i \sim \Unif([M])}(|\hat{Y}_i^{(t)}| > \lambda) \lesssim 2^{- \lambda}.
	\end{equation*}
\end{fact}
\begin{proof}
	Consider random set $S_k^{(t)} := \{ s \in S^{(t)} : \lsb(h_i(s)) \leq k \}$. 

    For $k = \log |S^{(t)}| - \lambda$ we have $\E |S_k| \simeq 2^{\lambda}$, and $\Var |S_k| \leq \E |S_k|$, hence $\P(|S_k| = 0) \leq 2^{-\lambda}$ by Chebyshev inequality, and therefore $\P(\hat{Y}_i^{(t)} > \lambda) < 2^{-\lambda}$.

    For the lower tail bound it is enough to consider Markov inequality: if $k = \log |S^{(t)}| + \lambda$, we have $\E |S_k| \leq 2^{-k}$, and $\P(|S_k| \geq 1) \leq 2^{-k}$.
\end{proof}

We will be interested in $Z_i^{(t)} := \log ( 2 + |\hat{Y}^{(t)}_i| )$ which is proportional to the number of bits necessary to write down deviation $Y_i^{(t)}$ from $\lg S^{(t)}$. The \Fact{subexp-tails} implies that $Z_i^{(t)}$ have doubly-exponential tail bounds, up to some rescaling: $\P_i(Z_i^{(t)} > c \lambda) \leq \exp(-e^{\lambda})$ for some constant $c$.

Let us take a sequence $t_1, t_2, \ldots t_R$, such that $|S^{(t_{k})}| = 2^k$, where $R = \Theta(\log n)$. 
We can now apply \LemmaR{sampler-lemma}, with $M := N$ and functions $g_k : [M] \to \mathbb{R}_+$ given by $g_k(i) = Z_i^{(t_k)} / c$ and $w = \Theta(\log \delta^{-1} + \log n)$. 

The final algorithm will be following: in the initialization phase, we choose a uniformly random string $S \in \{0, 1\}^s$ and store it. Consider now $\Gsamp_S : [w_1] \to [N]^{w_2}$, as in the statement of \LemmaR{sampler-lemma}, which for each value $t \in [w_1]$ yields a group of $w_2$ seeds for pairwise independent hash function from $[n]$ to $[n]$. For every such group, for example a group $\Gsamp_S(k)$, we store all the values $\{Y_i^{(t)} : i \in \Gsamp_s(k)\}$ in the compressed form: we store separately a median of all estimates within group, and the differences between $Y_i^{(t)}$ and aforementioned median. If at any point in time a size of the whole description of a given group in bits exceeds some $C_2 w_2$ we mark this group as broken and we stop updating it (where $C_2 = 2 C c$). Clearly, the total space complexity is bounded by $\Oh(s + C_2 w_2 w_1) = \Oh(s + w) = \Oh(\log n + \log \delta^{-1})$.
 
We claim that from the $C$-smallness condition, we can deduce that for a majority of $t \in [w_1]$, at all times both the median of all the estimates within group is close to the actual $|S^{(t)}|$, and the total space to store the whole group is bounded by $C_2 w_2$. If this is the case, then as an estimate for $|S^{(t)}|$ we can just report the median over groups $t \in [w_1]$ that are not marked as broken, of all the medians within a group of estimates $Y_i^{(t)}$, and the correctness of the algorithm follows.

To finish the argument, we need to show that every $C$-small group of estimators indeed yields a good approximation for $F_0$, and is stored succinctly at all times (i.e. never becomes marked as broken). Consider a $C$-small group $H \in [N]^{w_2}$. First, we will argue that at all times $t_k$, we have $\log (\median(i \in H: Y_i^{(t_k)})) = \log |S^{(t)}| \pm C_3$. Indeed, we know that on average over all $i \in H$, we have $\E_{i \in H} \log |2 + Y_i^{(t_k)}| < C$, therefore for at least $2/3$ fraction of $i \in H$, we have $\log |2 + \hat{Y}_i^{(t_k)}| \leq 3 C$, which means that for those $i$ we have $\hat{Y}_i^{(t_k)} \leq 2^{3C}$. This, together with the definition of $\hat{Y}$ implies the claim with $C_3 = 2^{3C}$. To argue that we are storing group $H$ using $\Oh(w_2)$ bits, let $M^{(t_k)}$ be the median of all $Y_i^{(t_k)}$ over $i$ in $H$. The space to store the group is given by $\Oh(\log \log n)$ bits to store the median within a group, and $\Oh(\sum_{i \in H} \log \left( 2 + |Y_i - M| \right))$ to store the rest. We have bound
\begin{equation*}
    \sum_{i \in H} \log \left(2 + |Y_i^{(t_k)} - M^{(t_k)}| \right) \leq \sum_{i \in H} \log \left( 2 + |Y_i^{(t_k)} - \log F_0^{(t_k)}|\right) + \sum_{i \in H} \log (1 + |M^{(t_k)} - \log F_0^{(t_k)}|) \leq \Oh(w)
\end{equation*}
where the first sum is bounded because of the $C$-smallness condition.

Finally, we have to say that if the group satisfies those two properties at all times $t_k$, then those properties are satisfied (with larger constants) at all time steps $t \leq T$. To see that, fix some $t$ between $t_k$ and $t_{k+1}$. Note that $Y_i^{(t)}$ is non-decreasing with respect to $t$, and we have
\begin{equation*}
	|Y_i^{(t)} - M^{(t)}| \leq |Y_i^{(t_{k+1})} - M^{(t_k)}|\leq |Y_i^{(t_{k+1})} - M^{(t_{k+1})}| + |M^{(t_{k+1})} - M^{(t_k})|.
\end{equation*}
Moreover, by triangle inequality 
\begin{equation*}
	|M^{(t_{k+1})} - M^{(t_k)}| \leq |M^{(t_{k+1})} - \log F_0^{(t_{k+1})}| + |M^{(t_{k})} - \log F_0^{(t_k)}| + |\log F_0^{(t_{k+1})} - \log F_0^{(t_k)}|
\end{equation*}
and each of those terms is bounded by constant.

This implies
\begin{align*}
    \sum_{i \in H} \log \left(2 + |Y_i^{(t)} - M^{(t)}|\right) 
    & \leq \sum_{i \in H} \left[\log \left(1 + |Y_i^{(t_{k+1})} - M^{(t_{k+1})}|\right) 
+  \log \left(1 + |M^{(t_k)} - M^{(t_{k+1})}|\right)\right] \\
&= \Oh(|H|) = \Oh(w_2).
\end{align*}

This completes the proof of the correctness of the algorithm --- at any step $t$, all $C$-small groups are not marked as broken, and all of them report a constant approximation. Strictly more than half of all the groups is $C$-small, hence the median of all groups that are still active, has to be a constant approximation to the quantity of interest as well.

\section{High accuracy regime \SectionName{high-accuracy}}

In this section we prove \TheoremR{high-accuracy}. As a building block we will use algorithm discussed in \cite[Section 3.2]{DBLP:conf/pods/KaneNW10}. In the Appendix we prove the following, qualitatively stronger bounds on the space complexity of their algorithm. The construction of the algorithm, and correctness analysis was already present in \cite{DBLP:conf/pods/KaneNW10} --- correctness can be also deduced from the discussion in \Section{strong-tracking}, where we discuss this algorithm in detail, and show stronger guarantees for a slight variation of it. Note that in the original paper the guarantees on the space complexity of this algorithm were proven when $\frac{1}{\varepsilon^2} > \log n$, as this could be assumed without loss of generality in their setting. For us, the scenario when $\frac{1}{\varepsilon^2} < \log n$ is relevant.

\begin{restatable}[]{theorem}{knwreduction}
	\TheoremName{knw-reduction}
	There is an algorithm $\mathcal{F}_\varepsilon$ which gives a $(1+\varepsilon)$-approximation to $F_0^{(t)}$ with probability at least $\frac{5}{6}$, assuming access to an oracle providing strong tracking of $F_0^{(t)}$ with constant factor approximation $C$, and oracle access to $\Oh(\log n + \poly\log(1/\varepsilon))$ additional random bits. The space usage of this algorithm at any given time $t$ (excluding random bits mentioned above), denoted by $W^{(t)}$, satisfies
	\begin{equation}
		\P(W^{(t)} > \frac{C_2}{\varepsilon^2}) \leq \varepsilon^4 \EquationName{small-eps}
	\end{equation}
	and
	\begin{equation}
		\P(W^{(t)} > \frac{\lambda}{\varepsilon^2}) \leq \exp(-e^{\Omega(\lambda)}).
		\EquationName{large-eps}
	\end{equation}
	Moreover for $t_1 < t_2$ such that $|S^{(t_1)}| \geq |S^{(t_2)}|/2$ we have 
	\begin{equation}
		W^{(t_1)} \leq W^{(t_2)} + \Oh(\frac{1}{\varepsilon^2}) \EquationName{timestops}
	\end{equation}
	for some universal constant $C$.
\end{restatable}

We will show how assuming this theorem we can prove \TheoremR{high-accuracy}, leveraging tools described in \Section{constant-approx}.

First of all, note that on a way to prove \TheoremR{high-accuracy}, we can assume without loss of generality that $\frac{1}{\varepsilon^2} < \log n$, for if it is not the case, we can just use $\log \delta^{-1}$ parallel repetitions of the KNW algorithm, to achieve $\delta$ failure probability with space $\Oh(\log n \log \delta^{-1} +  \frac{\log \delta^{-1}}{\varepsilon^2}) = \Oh(\frac{\log \delta^{-1}}{\varepsilon^2})$. In particular, this implies that the number of random bits used in \TheoremR{knw-reduction} is $\Oh(\log n + \poly \log \frac{1}{\varepsilon^2}) = \Oh(\log n)$.

We consider two separate cases, depending on relation between $\varepsilon$ and $\log n$. First, let us discuss case when $\varepsilon < \left(\frac{1}{\log n}\right)^{1/4}$. In this scenario, \Equation{small-eps} implies that for any specific position with probability $1 - \frac{C}{\log n}$ the total space consumption of a single instance of KNW algorithm use space $\Oh(\frac{1}{\varepsilon^2})$. Because of \Equation{timestops}, we can union bound only over positions for which $F_0^{(t)}$ grows by a factor of two (there are $\Oh(\log n)$ such positions), to ensure that with probability $\frac{5}{6}$ single instantiation of the algorithm uses space $\Oh(\frac{1}{\varepsilon^2})$ at all times.

We will use $\Oh(\log \delta^{-1})$ parallel instantiations of this algorithm. We use an algorithm which existence is guaranteed by \TheoremR{constant-approx} instantiated with failure probability $\delta$ to provide a strong-tracking oracle for all those implementations simultaneously. Instead of using $\log \delta^{-1}$ independent seeds across different instantiation of the algorithm $\tilde{F}_\varepsilon$, we consider the following standard pseudorandom object raising from random walks over explicit low degree expander graphs.

\begin{Definition}{\cite[Chapter 3]{Vadhan12}}
    A function $\Gamma : \{0, 1\}^s \times [w] \to [M]$ is \emph{$(\varepsilon, \delta)$-averaging sampler}, if for any function $f: [M] \to [0,1]$ and random variables $Y_i := \Gamma(U, i)$ for uniformly random $U$, we have
    \begin{equation*}
	    \P\left(\left|\sum_{i \leq w} f(Y_i) - \mu w\right| > \varepsilon w\right) < \delta,
    \end{equation*}
    where $\mu := \E_{Y \sim \Unif([M])} f(Y)$.

    A $(\varepsilon, \delta)$ sampler is called \emph{explicit} if $\Gamma$ can be computed in polynomial time in $s$ and $w$.
\end{Definition}
\begin{theorem}{\cite[Corollary 4.41]{Vadhan12}}
    \TheoremName{expander-ave-sampler}
    For every $\varepsilon, \delta$ there exist an explicit $(\varepsilon,\delta)$-averaging sampler, with the number of samples $w = \Oh_\varepsilon(\log \delta^{-1})$ and seed length $s = \log M + \Oh_\varepsilon(w)$, where $\Oh_\varepsilon$ notation hides constant depending on $\varepsilon$.

    Moreover $\Gamma$ can be computed using space $\Oh(s + w)$.
\end{theorem}

Consider $[M]$ to be the space of all possible random strings that were to be supplied to the algorithm $\mathcal{F}_\varepsilon$, and note that $\log M = \Oh(\log n)$. Let us fix an input stream, and condition on specific realization of the constant approximation tracking oracle (assuming that it succeeded --- we can bound the failure probability by $\delta$). For $k \in [M]$ let $\tilde{F}_0(k)$ be the approximation to $F_0^{(T)}$ reported by algorithm $\mathcal{F}_\varepsilon$ while supplied random string corresponding to $k$. We can define $f : [M] \to \{0, 1\}$, to be $f(k) := 1$ if and only if $|\tilde{F}_0(k) - F_0^{(T)}| < \varepsilon F_0^{(T)}$. Clearly, we have $\E_{k \sim \Unif([M])} f(k) > \frac{5}{6}$ --- this follows from the correctness guarantee for algorithm $\mathcal{F}_\varepsilon$.

Consider now $\Gamma$, a $(\frac{1}{6}, \delta)$ averaging sampler as in \TheoremR{expander-ave-sampler}. Except with failure probability $\delta$ over uniform random seed $S$ it will yield us a sequence $\Gamma_S(1), \ldots, \Gamma_S(w)$ of seeds, such that at least $\frac{2}{3} w$ amongst $\tilde{F}_0(\Gamma_S(1)), \ldots \tilde{F}_0(\Gamma_S(w))$ yields a good approximation to $F_0$. In this case, if we report median of all $\tilde{F}_0(k)$, it will be a valid answer.

The space of this algorithm is $\Oh(\log n + \log \delta^{-1})$ for the constant approximation oracle, $\Oh(\log n + \log \delta^{-1})$ for storing the seed to the averaging sampler, and $\Oh(\frac{\log \delta^{-1}}{\varepsilon^2})$ for storing all $w$ instantiations of $\mathcal{F}_\varepsilon$ algorithm. This yields total space complexity $\Oh(\log n + \log \delta^{-1})$. The failure probability of each one of the three phases is bounded by $\delta$, hence the total failure probability is bounded by $3 \delta$, and the result follows by rescaling $\delta$ by a constant factor.

Let us now turn our attention to the analysis of the second case, where $\varepsilon > \left(\frac{1}{\log n}\right)^{1/4}$. In this case, the proof will make use of \Equation{large-eps}, and mimic the proof of \TheoremR{constant-approx}. Note that in this regime of parameters, we can assume without loss of generality, that $\log \delta^{-1} \geq C \sqrt{\log n}$, as otherwise we could take $\delta'$ such that $\log \delta'^{-1} = \log \delta^{-1} + C\sqrt{ \log n }$, and the additional $\frac{C \sqrt{\log n}}{\varepsilon^2}$ term in the space complexity will be dominated by $\Oh(\log n)$ term anyway.

First of all, by naive failure probability amplification, after adjusting other constant in \TheoremR{knw-reduction}, we can actually assume that failure probability of this algorithm is small constant $c_0$. We will apply \LemmaR{sampler-lemma} where the universe $[M]$ is given by $M=2^{r}$, with $r$ being the number of random bits accessed by this new adjusted algorithm (in particular $\log M =  \Theta(\log n)$). 

Let us take a sequence $t_1, \ldots t_{R-1}$ where $t_1 = 0$, and each $t_j$ for $j > 1$ is smallest such that $F_0^{(t_j)} \geq 2 F_0^{(t_{j-1})}$. Clearly, $R = \Oh(\log n)$.

We will use $g_1, \ldots g_{R}$ to be given by $g_i(m) := \frac{\varepsilon^2}{C} W^{(t_i)}_m$ for $i \leq R - 1$, where by $W^{(t_i)}_m$ we denote the space consumption of the instantiation of the algorithm described in \TheoremR{knw-reduction} with random bits given by $m \in [M]$. Finally, we pick $g_{R}(m)$ to be $0$ if the instance of algorithm corresponding to random bits $m$ succeeds to provide $1+\varepsilon$ approximation, and some large $C_0 \geq 3 C$ if it fails. Given that failure probability $c_0$ is small enough depending on $C_0$, we can ensure that this function $g_R$ indeed satisfy doubly-exponential tail bounds. For all previous functions $g_{1}, \ldots g_{R-1}$, doubly exponential tail bounds are guaranteed by \Equation{large-eps}. Finally, we can apply \LemmaR{sampler-lemma} with $w = \Theta(\log \delta^{-1})$ --- we assumed that $\log \delta^{-1} = \Omega(\sqrt{\log n})$, so the assumptions of this lemma are satisfied.

The sampler $\Gsamp$ guaranteed by \LemmaR{sampler-lemma} returns a sequence of groups of estimators, such that most of those groups are $C$-small (except with small failure probability $\delta$ over the choice of the seed). We wish to argue that if the sampler succeeds (i.e. most of the reported groups is $C$-small), then the algorithm will use small space, and will correctly return $(1+\varepsilon)$ approximation for the number of distinct elements. As in \Section{constant-approx}, we can discard any group for which the space consumption becomes too large over the course of algorithm, hence the total space is $\Oh(\frac{w}{\varepsilon^2} + \log n) = \Oh(\frac{\log \delta^{-1}}{\varepsilon^2} + \log n)$. By $C$-smallness condition restricted to functions $g_1, \ldots g_{R-1}$ and \Equation{timestops} majority of the groups (all $C$-small groups) are never discarded in this way --- the argument for this is identical as in the proof of \TheoremR{constant-approx}. We need to argue, that reporting median of all the medians within surviving groups indeed yields $(1+\varepsilon)$-approximation to the number of distinct elements. This is guaranteed by $C$-smallness condition applied to function $g_R$ --- indeed, for large enough $C_0$ we can ensure that any $C$-small group of estimators have at least $\frac{2}{3}$ fraction of estimators reporting value that is within $(1+\varepsilon)$ to the actual answer. 

\section{Strong tracking of distinct elements\SectionName{strong-tracking}}

In this section we prove \TheoremR{strong-tracking}. Let us first state a technical lemma essential in the argument. After stating this lemma we will show how, together with \LemmaR{4-wise-indep-random-walk}, those two imply \TheoremR{strong-tracking}. The rest of this section will be devoted to proving those lemmas.

\begin{Definition}
    For a finite universe $[K]$, let $\phi_K : [K]^* \to \bN$ be given by
    \begin{equation*}
        \phi_K(r_1, \ldots r_s) := \#\{ j \in K : \exists i, r_i = j\}.
    \end{equation*}

    Moreover, let $\Phi_K : \bN \to \bN$ be given by
    \begin{equation*}
        \Phi_K(t) := \E \phi(X_1, \ldots X_t)
    \end{equation*}
    over $X_1, \ldots X_t$ uniformly random in $[K]$ and independent.

    We will skip the index $K$, when the underlying universe is clear from context.
\end{Definition}

Function $\phi$ counts the number of non-empty bins after throwing balls into $K$ bins, and the following lemma states that if we track then number of non-empty bills while throwing balls at random it stays close to the expectation. A lemma like this would be much simpler corollary of Doobs Martingale inequality, if variables $X_1, \ldots X_R$ were known to be fully independent.
\begin{lemma}
    \LemmaName{balls-and-bins}
	Consider a sequence $X_1, \ldots X_R \in [K]$ $q$-wise independent for some $q = \Theta(\polylog R)$ where marginal distribution of each $X_i$ is uniform over $[K]$, and $R \leq K/20$. Then
    \begin{equation*}
        \sup_{t \leq R} |\phi(X_1, \ldots X_t)  - \Phi(t)| = \Oh(\sqrt{R})
    \end{equation*}
    with probability $3/4$.
\end{lemma}

The following fact will also be useful
\begin{fact}[\cite{DBLP:conf/pods/KaneNW10}]
    \FactName{phi-Lipschitz}
     We can calculate $\Phi$ exactly as $\Phi_K(t) = t \left(1 - \left( 1 - \frac{1}{t}\right)^n\right)$ for $t > 0$, and $\Phi_K(0) = 0$.

    Moreover for all $0 \leq \alpha_1, \alpha_2 \leq K/20$, we have $0.9 |\alpha_2 - \alpha_1| \leq |\Phi(\alpha_2) - \Phi(\alpha_1)| \leq |x_2 - x_1|$.

    Finally, for $t < P/20$, we have $\Var(\phi(X_1, \ldots X_t)) \leq \Phi(t) \leq t$.
\end{fact}

\begin{proof}[Proof of \TheoremR{strong-tracking}]
	First we will discuss how, given an upper bound $P$ on the number of distinct elements, we can analyze a variant of the algorithm in \cite{DBLP:conf/pods/KaneNW10} to argue, that in fact at all times $t$ it provides a $\pm \varepsilon K$ additive approximation to $|S^{(t)}|$, without any additional space blowup. This can be used to say that after amplifying the failure probability to $\delta/\log n$, by union bound over all positions where $|S^{(t)}|$ grows by a factor of two, we can obtain strong tracking guarantee with failure probability $\delta$.
	
	Take $P = \frac{100}{\varepsilon^2}$, and let us consider a 8-wise independent hash-function $h_1 : [n] \to [n]$ (8-wise independence here is used to get correct bounds on the space complexity in \TheoremR{knw-reduction}) , and random sets $S_k^{(t)} := \{ s \in S^{(T)} : \lsb(h_1(s)) \geq 2^k \}$ as previously. Let us consider in addition a pairwise independent hash function $h_3 : [n] \to [P^2]$, and finally a $\polylog(P)$-wise independent hash function $h_4 : [P^2] \to [P]$. Define $h_2 : [n] \to [P]$ to be the composition $h_2 := h_4 \circ h_3$.

    In the Appendix~\ref{sec:appendix} it is discussed how, given oracle access to constant factor strong tracking, we can maintain a sketches of size $\Oh(\frac{1}{\varepsilon^2})$ on average (with some small constant probability of failure), such that we can recover $|h_2(S_k)|$ for any $k$ at any point of the stream.

    Let us fix any $K < n$, and let $k$ be such that $2^{-k} K \approx \frac{1}{10 \varepsilon^2}$. We wish to show that with probability $\frac{9}{10}$ we have
    \begin{equation}
        \forall t\leq T_0, |\Phi^{-1}(|h_2(S^{(t)}_k)|) 2^k| = |S^{(t)}| \pm \Oh(\varepsilon K)
        \EquationName{weak-tracking}
    \end{equation}
    where $T_0$ is such that $|S^{(T_0)}| = K$.

    If this were true, we could repeat the construction $\Oh(\log \log n + \log \delta^{-1})$ times, to amplify success probability for the median estimator to $1 - \frac{\delta}{\log n}$, and use a union bound to ensure that \Equation{weak-tracking} is satisfied for all $k$ simultaneously. This can kind of amplification can be implemented exactly as described in \Section{high-accuracy}.

    Given access to strong tracking oracle with constant failure probability, we know which set $|S_k|$ to use, at any given time, to estimate $|S^{(t)}|$, as above.

    We only need to show that \Equation{weak-tracking} indeed holds with large constant probability. We can assume without generality that $|S^{(t)}| = t$, i.e. all the elements in the input stream are distinct.

    First of all, we will show that $\forall t \leq T_0,\ |S_k^{(t)}|2^k = |S^{(t)}| \pm \varepsilon K$. Indeed, note that if we take $X_k := |S_k^{(t)}|2^k - |S^{(t - 1)}| 2^{k} - 1$, we can see that $X_k$ are 4-wise independent (because hash function $h_1$ was assumed to be 4-wise independent), and satisfy $\E X_k = 0$, $\E X_k^2 \approx 2^k$. By applying \LemmaR{4-wise-indep-random-walk} we see that $\sup_{t \leq T_0} \left|2^k |S^{(t)}_k| - t\right| = \Oh(\sqrt{ 2^{k} T_0}) = \Oh( \varepsilon K )$.

    By birthday paradox, with probability $9/10$ we have $|h_3(S_k^{(t)})| = |S_k^{(t)}|$, i.e. function $h_3$ has no collisions in the part of the stream of interest. Moreover, by \LemmaR{balls-and-bins}, conditioned on $S_k^{(t)}$ we have that with high probability at all times $t$ in the range of interest that $|h_2(S_k^{(t)})| = \Phi(|S_k^{(t)}|) \pm \Oh(\frac{1}{\varepsilon})$. This together with \Fact{phi-Lipschitz}, implies that $\Phi^{-1}(|h_2(S_k^{(t)})|)$ yields at all times a good approximation of $S_k^{(t)}$, and by composing with the previous argument, this shows \Equation{weak-tracking}
\end{proof}

The rest of this section will be devoted to proving \LemmaR{4-wise-indep-random-walk} and \LemmaR{balls-and-bins}.


\subsection{Kolmogorov inequalities with bounded independence}
Here we will show \LemmaR{4-wise-indep-random-walk}. Let us first discuss a similar lemma under fourth moment assumptions, but crucially without any independence assumptions on the increments. This lemma we will use to control deviations of pseudorandom versions of Doobs martingale in the proof of \LemmaR{balls-and-bins}. The following proof is basically present in \cite{BravermanCINWW17}, although with different statement of the lemma --- understanding it will be helpful in understanding much more delicate proof of \LemmaR{4-wise-indep-random-walk}, where we do not have control over fourth moments.
\begin{lemma}
		\LemmaName{doobs}
		Let $X_1, X_2, \ldots X_T$ be collection of random variables with $\E X_i = 0$, and let $Y_j := \sum_{i \leq j} X_j$. If for any $i, j$ we have $\E(Y_i - Y_j)^4 \lesssim (i - j)^2$, then
		\begin{equation*}
				\P(\sup_{i} Y_i \geq \sqrt{T} \lambda) \lesssim \lambda^{-4}
		\end{equation*}
\end{lemma}
\begin{proof}
		We assume without loss of generality that $T = 2^n - 1$, define $A_0 = \{0\}$, and in general $A_k := \{ j 2^{n-k} : j \in \{0, 1, \ldots 2^k - 1\}\}$, let us moreover define $\delta_k := 2^{n-k}$. 

		For each $k>1$ we have $\Pr(\exists j \in A_k - A_{k-1}, |Y_j - Y_{j - \delta_k}| > \lambda 2^{k/3} \sqrt{\delta_k}) \lesssim |A_k| (\lambda 2^{k/3})^{-4} \leq \lambda^{-4} 2^{-k/3}$. Note that if $j \in A_{k} - A_{k-1}$ then $j - \delta_k \in A_{k-1}$.

		Except with probability $\sum_k \lambda^{-4} 2^{-k/3} \lesssim \lambda^{-4}$ all of those events happen simultaneously for all $k$. In this case, by applying triangle inequality we have 
		\begin{equation*}
				\forall j |Y_j| \leq \sum_k 2^{k/3}\sqrt{\delta_k}\lambda = \sum_k 2^{n/2 - k/6}\lambda \lesssim \sqrt{T} \lambda.
		\end{equation*}
\end{proof}

\begin{proof}[Proof of \LemmaR{4-wise-indep-random-walk}]
    Let us assume without loss of generality that $T = 2^n$. Take $A_0 := \{0, T\}$, and for $k \leq n$, take $A_k := \{ j 2^{n-k} : j \in \{0, 1, \ldots 2^k\}\}$.

    For every $t \in A_{k} \setminus A_{k-1}$, it has two neighbours in $A_{k-1}$: those are $t + \delta_k, t - \delta_k \in A_k$, where $\delta_k := 2^{n-k+1}$.

    Observe that $\E_{X \sim \mathcal{D}} (S_t - S_{t - \delta_k})^2 \leq \delta_k$, and similarly $\E_{X \sim \mathcal{D}} (S_t - S_{t + \delta_k})^2 \leq \delta_k$, and moreover $\E_{X \sim \mathcal{D}} (S_t - S_{t - \delta_k})^2 (S_t - S_{t + \delta_k})^2 \leq \delta_k^2$ --- this can be shown by expanding both sums $(S_t - S_{t - \delta_k})$ and $(S_{t + \delta_k} - S_t)$ --- because of $4$-wise independence and $\E X_i X_j = 0$ for $i\not=j$, we have $\E (S_t - S_{t-\delta_k})^2 (S_{t+\delta_k} - S_t)^2 = \E \sum_{1 \leq i \leq \delta_k} \sum_{1 \leq j \leq \delta_k} X_{t - \delta_k + i}^2 X_{t + j}^2 \leq \delta_k^2$.
    
    Those bounds on second moments, together with Markov inequality, yield a bound
    \begin{equation*}
        \P( \min \{ | S_t - S_{t - \delta_k} |, | S_t - S_{t + \delta_k}| \} > \lambda \sqrt{\delta_k}) \leq \P( (S_t - S_{t - \delta_k})^2 (S_t + S_{t + \delta_k})^2 > \lambda^4 \delta_k^2) \leq \frac{1}{\lambda^4}
    \end{equation*}

    For $k \geq 1, t \in A_k$ define a bad event $E_{k, t}$ to be $\min\{|S_t - S_{t-\delta_k}|, |S_t - S_{t+\delta_k}|\} > \gamma \sqrt{\delta_k} 2^{k/3}$. We have $\P(E_{k, t}) \leq \frac{2^{-4k/3}}{\gamma^4}$. Since $|A_{k} \setminus A_{k-1}| = 2^{k-1}$, for every $k$ we have $\P(\exists t\in A_k\setminus A_{k-1}: \, E_{k, t}) \leq \frac{2^{-k/3}}{\gamma^4}$, and finally by taking union bound over all $k$, we have $\P(\exists k, t\in A_k\setminus A_{k-1}:\, E_{k,t}) \lesssim \frac{1}{\gamma^4}$.

    We now claim, that if no of the events $E_{k,t}$ happened, we have 
    \begin{equation}
        \sup_{t\leq T} |S_t| \leq |S_T| + \gamma \sum_{k \leq n} \sqrt{\delta_k} 2^{k/3}.
        \EquationName{chain-induct}
    \end{equation}
    Before we prove that, let us observe that $\sum_{k \leq n} \sqrt{\delta_k} 2^{k/3} \leq \sqrt{T} \sum_{k \leq n} 2^{-k/6} = \Oh(\sqrt{T})$.

    We will show the following fact by induction over $k_0$
    \begin{equation*}
        \sup_{t \in A_{k_0}} S_t \leq S_T + \gamma \sum_{k\leq k_0} \sqrt{\delta_k} 2^{k/3}.
    \end{equation*}

    Clearly for $k_0 = 0$ this is satisfied. Moreover, for $k_0 > 1$, if $t\in A_{k_0} - A_{k_0 - 1}$, we have some $\tilde{t} \in A_{k_0 - 1}$, with $|S_t - S_{\tilde{t}}| \leq \gamma \sqrt{\delta_{k_0}} 2^{k_0/3}$ --- this follows from the fact that event $E_{k_0, t}$ did not happen. This yields
    \begin{equation*}
        |S_t| \leq |S_{\tilde{t}}| + |S_t - S_{\tilde{t}}| \leq \sup_{t^* \in A_{k_0 - 1}} |S_{t^*}| + \gamma \sqrt{\delta_{k_0}} 2^{k_0/3}
    \end{equation*}

    On the other hand, $\E S_T^2 \leq T$, hence by Chebyshev inequality $\P(|S_T| > \sqrt{T} \gamma) \leq \frac{1}{\gamma^2}$. This, together with inequality~\Equation{chain-induct} yields a tail bound $\P(\sup_{t \leq N} S_t > K \sqrt{N} \gamma) = \frac{1}{\gamma^2} + \Oh(\frac{1}{\gamma^{4}})$ for some universal constant $K$; after changing $\gamma$ by this constant factor, we conclude the statement of the lemma.
\end{proof}




\subsection{Pseudorandom balls and bins}
For the proof of the \LemmaR{balls-and-bins} we will need following statement of the Chernoff inequality

\begin{theorem}[Chernoff bound {\cite[Lemma 2.3]{BellareR94}}]
\TheoremName{chernoff}
If $X_1, \ldots X_n$ are $r$-wise independent random variables satisfying $0 < X_i < 1$ almost surely, with $\mu := \E \sum X_i$ then for $\lambda > 2$ we have
\begin{equation*}
    \P(\sum X_i > \lambda \mu) < \exp(-\Omega(\max\{ \mu \lambda, r \} \log \lambda))
\end{equation*}
\end{theorem}

The strategy for the proof of~\LemmaR{balls-and-bins} is following. For random variables $X_1, \ldots X_R$ as in the statement of the lemma, we would like to control process $S_t = \E_{X'} \phi(X_1, \ldots X_t, X'_{t+1}, \ldots X'_R)$, and specifically the deviations $\sup_{t} |S_t - \E S_t|$. If variables $X_i$ were truly independent, that would be given by the Doobs martingale inequality. As they are not, we first show that $\phi$ can be approximated in appropriate sense by a low degree polynomial $\hat{\phi}$. Then we control analogous process $\hat{S}_t$ defined on top of the approximation $\hat{\phi}$ --- in order to do this, we observe that low moments of the increments $\hat{S}_j - \hat{S}_i$ are bounded --- they are expectations of low degree polynomials of input variables $X_i$, hence they are the same as if the variables were truly independent --- and in such a case we can use known results about martingales to reach the desired conclusion.

\begin{lemma}
    \LemmaName{polynomial-approximation}
	For any $P$ and $R \leq P/20$, there exists a polynomial $\hat{\phi} : \{0, 1\}^{\log P \times R} \to \bR$ of degree $ \Oh(\log^2 P)$ with integer coefficients, such that for every distribution $X_1, \ldots X_R$ which is at least $r = \poly(\log P)$-wise independent and with marginal distribution of each $X_i$ being uniform, we have
    \begin{equation*}
	\|\hat{\phi}(\bar{X}_1, \ldots \bar{X}_R) - \phi(X_1, \ldots X_R)\|_p = o(1)
    \end{equation*}
    for any $p \lesssim \log \log P$. In particular
    \begin{equation}
	\P(\hat{\phi}(\bar{X}_1, \ldots \bar{X}_R) \not= \phi(X_1, \ldots X_R)) = o(1).
        \EquationName{rarely-different}
    \end{equation}

    Above, $\bar{X}_i$ denotes the binary representation of $X_i$.
\end{lemma}
\begin{proof}
    Consider polynomials $\EQ_k : \bR^{\log P} \to \bR$, such that $\EQ$ restricted to the hypercube $\{0, 1\}^{\log P}$ has values $\{0, 1\}$ and it takes value $1$ only for argument $\bar{k}$ (for a number $k \in [P]$, we write $\bar{k}\in \{0,1\}^{\log P}$ to be the binary representation of $k$). There is such a multilinear polynomial of degree $\log P$.

    Consider moreover polynomial $I_0 : \bR \to \bR$ of degree $d = \Theta(\log P)$ defined as
    \begin{equation*}
	    I_0(x) := 1 - \sum_{0 \leq i \leq d} (-1)^{i} \binom{x}{i}
    \end{equation*}

    Let us observe that for $x \in \bN$, with $x \leq d$ we have
    \begin{equation*}
        I_0(x) = \left\{ \begin{array}{cl} 0 \qquad & \text{for }x=0 \\
            1 \qquad & \text{otherwise}
        \end{array}\right.
    \end{equation*}
    and moreover $|I_0(x)| \lesssim {\binom{x}{d + 1}}$ for any $x > 0$. 

    Let us define now 
    \begin{equation*}
        \hat{\phi}(\bar{X}_1, \ldots \bar{X}_R) := \sum_{k \in [P]} I_0(\sum_{j \leq R} \EQ(\bar{X}_j, k))
    \end{equation*}
	and note that this is a polynomial of degree $d \log P = \Oh(\log^2 P)$.
    
    Given an instantiation of random variables $X_1, \ldots X_R$, we will take $B_i := \#\{ j : X_j = i \}$ defined for every $i \in [P]$, and moreover we will define $M(X_1, \ldots X_R) := \max_{i \in [P]} B_i$.

    We claim that $\hat{\phi}(\bar{X}_1, \ldots \bar{X}_R) = \phi(X_1, \ldots X_R)$ as long as $M(X) \leq d$, and for $M(X) > d$ we have $|\hat{\phi}(X_1, \ldots X_R) - \phi(X_1, \ldots X_R)| < N \binom{M(X)}{d} < N \exp(c_0 d \ln \frac{M(X)}{d})$ for some constant $c_0$. This yields
    \begin{align}
        \E |\hat{\phi} - \phi|^q & \leq 
        \sum_{k > \log d} \P(2^{k} < M(X) < 2^{k+1}) \E\left[ N^q 
        \exp(c_0 q d \log \frac{M(X)}{d}) \middle| 2^{k} < M(X) < 2^{k+1}\right] \nonumber \\
        & \leq \sum_{k > \log d} \P(M(X) > 2^k) \exp(c_0 q d \log \frac{2^{k}}{d} + q \ln N). \EquationName{q-norm-phi}
    \end{align}

    We can bound the tail probabilities of $M(X)$ as follows
    \begin{align}
        \P(M(X) > \lambda) & \leq \sum_{j \leq R} \P(B_j > \lambda) \nonumber \\
        & \leq R \P(B_1 > \lambda) \nonumber \\
        & \leq R \exp(-C \min(\lambda, r) \log \lambda) \EquationName{m-tail-bound}
    \end{align}
    where the last inequality follows from the Chernoff bound \TheoremR{chernoff}, because $\E B_i < 1/20$.

    \Equation{m-tail-bound} together with \Equation{q-norm-phi} yields
    \begin{equation*}
        \E |\hat{\phi} - \phi|^q  \leq
	\sum_{k > \log d} \exp(-C(k \min(2^k, r)) + d q \log \frac{2^k}{d} + q \log R) 
    \end{equation*}
    The exponents in this sum are quickly decaying, so the whole sum is of the same order as the first term, namely
    \begin{equation*}
        \exp( - \Theta(d \log d) + \Theta(d q) + q \log R)
    \end{equation*}
	if we pick $r > d^{\Oh(1)}$.

    Hence, for $q \ll \log d$ and $d \gg \log R$ we have
    \begin{equation*}
        \E |\hat{\phi} - \phi|^q  \leq \exp(-\Omega(d \log d)).
    \end{equation*}

\end{proof}

We will now show that for any distribution with enough independence, specific types of random walks associated with functions $\phi$ and $\hat{\phi}$ stay close to each other with good probability. If variables $X_1, \ldots X_R$ are uniform and independent, processes $S_t$ described below are just Doobs martignales associated with function $\phi$, that were used to show correctness of the algorithm in the random-oracle model.

\begin{lemma}
    \LemmaName{close-martingales}
    Consider $\hat{\phi}$ to be the polynomial from \LemmaR{polynomial-approximation}, and let $X_1, \ldots X_R \in \Sigma$ be a sequence of $k$-wise independent random variables such that marginal distribution of each $X_i$ is uniform, where $k = \poly(d)$.

    Consider $S_t := \E_{X'} \phi(X_1, \ldots X_t, X'_{t+1}, \ldots X'_R)$, and
    $\hat{S}_t := \E_{X'} \hat{\phi}(X_1, \ldots X_t, X'_{t+1}, \ldots X'_R)$, where $X'$ are independent random variables, distributed uniformly over $\Sigma$. Then
    \begin{equation}
        \P(\exists t,\ |\hat{S}_t - S_t| > \lambda) \lesssim \frac{R}{\lambda^2}.
    \end{equation}
\end{lemma}
\begin{proof}
	For single $t$ we have 
	\begin{align*}
		\|\hat{S}_t - S_t\|_2^2 &\leq \E_X | \E_{X'} \phi(X_1, \ldots X_t, X'_{t+1}, \ldots, X'_R) \\
		& \phantom{\leq} - \hat{\phi}(X_1, \ldots X_t, X'_{t+1}, \ldots X'_R)|^2 \\
		& \leq \E_{X, X'} |\phi(X_1, \ldots X_t, X'_{t+1}, \ldots X'_R) \\
		& \phantom{\leq}- \hat{\phi}(X_1, \ldots X_t, X'_{t+1}, \ldots X'_R)|^2\\
		& \lesssim 1,
	\end{align*}
	where the last inequality follows from \LemmaR{polynomial-approximation}.

    Hence, $\P(|\hat{S}_t - S_t| > \lambda) \lesssim \frac{1}{\lambda^2}$, and by union bound

    \begin{equation*}
        \P(\exists t,\ |\hat{S}_t - S_t| > \lambda) \lesssim \frac{R}{\lambda^2}
    \end{equation*}
\end{proof}

\begin{lemma}
    \LemmaName{bounded-increments}
    For $X_1, \ldots X_R$ and $\hat{S}_i$ defined as in the \LemmaR{close-martingales}, if $\Delta_i := \hat{S}_{i} - \hat{S}_{i-1}$, then
    \begin{align}
		\E \Delta_i & = 0 \\
		\E (\hat{S}_i - \hat{S}_j)^4 & \lesssim (i - j)^2
    \end{align}
\end{lemma}
\begin{proof}
		Note that all the expressions in the statement of the lemma are expectations of polynomials of degree at most $4d$ in variables $X_i$, and variables $X_1, \ldots X_R$ are $r$-wise independent for $r > 4d$. We can without loss of generality prove this theorem assuming that $X_1, \ldots X_R$ are instead independent uniform random variables. 

		In that case $\Delta_i$ is a sequence of increments of a Doob's martingale, and therefore $\E \Delta_i = 0$. Similarly, since $\hat{S}_i$ is Doobs martingale we can apply Lemma~\ref{lem:azuma} to deduce that $\E (\hat{S}_i - \hat{S}_j)^4 \lesssim (i - j)^2$.
\end{proof}

\begin{corollary}
		\CorollaryName{bounded-hats-tails}
    For $\hat{S}_k$ defined as above, we have $\P(\sup_{k \leq R} |\hat{S}_k - \hat{S}_0|\geq \lambda) \lesssim \frac{R}{\lambda^2}$.
\end{corollary}
\begin{proof}
    Folows from \LemmaR{bounded-increments} and \LemmaR{doobs}.
\end{proof}

\begin{corollary}
    \CorollaryName{polynomial-tails}
    For $S_t$ defined as in \LemmaR{close-martingales}
    \begin{equation*}
        \P(\sup_t\, |S_t - S_0| \geq \lambda) \lesssim \frac{R}{\lambda^2}
    \end{equation*}
\end{corollary}
\begin{proof}
		Since $|S_t - S_0| \leq |S_t - \hat{S}_t| + |\hat{S}_t - \hat{S}_0| + |\hat{S}_0 - S_0|$, clearly we have
    \begin{equation*}
        \P(\sup_t\, |S_t - S_0| \geq \lambda) \leq \P(\exists t,\ |S_t - \hat{S}_t| \geq \lambda/3) + \P(\sup_t\, |\hat{S}_t - \hat{S}_0| \geq \lambda/3)
    \end{equation*}
    and the claim follows from \Corollary{bounded-hats-tails} and \LemmaR{close-martingales}. 
\end{proof}
\begin{Remark}
    \RemarkName{smaller-bins}
	If $\hat{X}_{t+1}, \ldots \hat{X}_R$ are all uniform and independent, then for any setting of variables $X_1, \ldots X_t$ we have
    \begin{equation*}
	    \E_{\hat{X}} \phi(X_1, \ldots X_t, \hat{X}_{t+1}, \ldots \hat{X}_R) = \Phi\left[ \Phi^{-1}(\phi(X_1, \ldots X_t)) + R - t\right].
    \end{equation*}
\end{Remark}
We are finally ready to prove the last technical lemma, stating that for bounded-wise independence balls-and-bins experiment, the number of non-empty bins stays close to its expectation at all times.
\begin{proof}[Proof of \LemmaR{balls-and-bins}]

	By \Corollary{polynomial-tails} we conclude that with probability $\frac{9}{10}$ we have
	\begin{equation*}
		\forall t, \E_{X'} \phi(X_1, \ldots X_t, X'_{t+1}, \ldots, X'_R) - \Phi(R) \leq \Oh(\sqrt{R}).
	\end{equation*}

	Using bi-Lipschitz properties of $\Phi_R$ (\Fact{phi-Lipschitz}), we deduce that equation above imply
	\begin{equation*}
		\forall t, \Phi^{-1}(\E_{X'} \phi(X_1, \ldots X_t, X'_{t+1}, \ldots. X'_R)) = R \pm \Oh(\sqrt{R}).
	\end{equation*}
	Applying \RemarkR{smaller-bins}, we deduce
	\begin{equation*}
		\forall t, \Phi^{-1}(\phi(X_1, \ldots X_t)) = t \pm \Oh(\sqrt{R})
	\end{equation*}
	and finally, again using bi-Lipschitz continuity of $\Phi$, we deduce
	\begin{equation*}
		\forall t, \phi(X_1, \ldots, X_t) = \Phi(t) \pm \Oh(\sqrt{R})
	\end{equation*}
\end{proof}

\section{Strong tracking lower bound \SectionName{lower-bound}}

In this chapter we prove \TheoremR{lower-bound} --- $\Omega(\frac{\log \log n}{\varepsilon^2})$ lower bound for strong tracking of distinct elements. To this end, we introduce concept of $T$-game --- model of communication-complexity game tailored to the lower bound in question. 

\begin{Definition}[$T$-game]
    For any relation $\mathcal{R} \subset \{0,1\}^n \times \{0,1\}^n \times \Sigma$, we consider $T$-game $T(\mathcal{R}, k)$ with $k$-rounds, to be communication problem with two parties, Alice and Bob defined as follows. In each round of the game
    \begin{itemize}
        \item Alice receives her input $x_k \in \{0,1\}^n$, and Bob receives his input $y_k \in \{0,1\}^n$.
        \item Alice receives Bobs input $y_{k-1}$ from the previous round, and Bob observes Alices input $x_{k-1}$ from the previous round.
        \item Alice and Bob can observe private random coins $r^1_k, r^2_k \in \{0, 1\}^*$.
        \item Alice can send a message $a_k$ to Bob that depends on all her observations.
        \item Bob reports to the judge his output $z_k \in \Sigma$.
        \item Bob can send a message $b_k$ to Alice.
    \end{itemize}

    We say that protocol $P$ \emph{succeeds} on input $(x_1, y_1), \ldots (x_k, y_k)$ and random coins $\left( (r_i^1, r_i^2) \right)_{i\in [k]}$ if $\forall_k (x_k, y_k, z_k) \in \mathcal{R}$. For any protocol $P$ by Alice and Bob, we define complexity $C(P)$ of the protocol to be the largest length of $a_k$, or $b_k$ sent by any party.

    For a distribution $\mu$ over pair of strings $\{0, 1\}^n \times \{0, 1\}^n$, let $\mathcal{P}_{\mu, \delta}$ be the set of all protocols that succeed with probability $1-\delta$, given as input sequence of independent samples $(x_1, x_2), \ldots (x_k, y_k) \sim \mu$. We define
    \begin{equation*}
        D_{\mu, \delta}(T(\mathcal{R}, k)) := \inf_{P \in \mathcal{P}_{\mu, \delta}} C(P)
    \end{equation*}
\end{Definition}

\begin{Definition}
    For relation $\mathcal{R} \subset \{0, 1\}^n \times \{0, 1\}^n \times \Sigma$, we will denote by $D_{\mu, \delta}^{\rightarrow}(\mathcal{R})$ the one-way deterministic communication complexity of $\mathcal{R}$ under distribution $\mu$ of inputs for Alice and Bob.
\end{Definition}

The following lemma connects complexity of $T$-game based on relation $\mathcal{R}$, with one-way communication complexity of the relation $\mathcal{R}$ itself.

\begin{lemma}
	\LemmaName{T-to-one-way}
	For every protocol for a $k$-round $T$-game with failure probability $\delta$, over independent samples distributed according to $\mu$ and complexity $C(P)$, there is a one-way communication protocol for a distribution $\mu$ with communication complexity $C(P)$ and failure probability $\delta/k$. Formally, for every relation $\mathcal{R}$ the following inequality holds
    \begin{equation*}
        D_{\mu, \delta/k}^\rightarrow(\mathcal{R}) \leq D_{\mu, \delta}(T(\mathcal{R}, k)).
    \end{equation*}
\end{lemma}
\begin{proof}
    Consider fixed protocol $P$ with $C(P) \leq S$. By standard averaging argument we can assume that $P$ is a deterministic protocol.
    
    Consider event $A_t$ given by $(x_t, y_t, z_t)\in \mathcal{R}$. $A_t$ --- such an event depends only on $\{(x_s, y_s)\}_{s \leq t}$. We have $\delta \geq \P(\bigvee_t \lnot A_t) = \sum_t \P(\lnot A_t | \bigwedge_{s < t} A_s)$, and therefore there is $t_0$ for which $\P(\lnot A_{t_0} | \bigwedge_{s < t_0} A_s) \leq \frac{\delta}{k}$. In particular, there exists $(\hat{x}_1, \hat{y}_1), \ldots (\hat{x}_{t_0-1}, \hat{y}_{t_0-1})$, such that
    \begin{equation*}
        \P_{(x_{t_0}, y_{t_0}) \sim \mu}(\lnot A_{t_0} | \forall i<t_0,\ (x_i, y_i) = (\hat{x}_i, \hat{y}_i)) \leq \frac{\delta}{k}
    \end{equation*}

    Now, Alice and Bob can fix those $(\hat{x}, \hat{y})$, and use the restriction of protocol $P$ to the $k$-th round as a single round one way communication protocol for $\mathcal{R}$. As described above, failure probability of this protocol is bounded by $\frac{\delta}{k}$.
\end{proof}

In what follows we will use $T$-games associated with following relation.

\begin{Definition}[Approximate distinct elements relation]
	We define relation $F_0^\varepsilon \subset \{0, 1\}^n \times \{0, 1\}^n \times \bZ$, to be $F_0^\varepsilon = \{ (x,y,z) : (1 - \varepsilon) |x \lor y| \leq z \leq (1 + \varepsilon) |x \lor y|\}$.
\end{Definition}

The one-way communication complexity of this relation, in the low failure probability range, can be lower bounded as follows.
\begin{theorem}[\cite{DBLP:journals/talg/JayramW13}]
	\TheoremName{jayram-woodruff}
    For every $\varepsilon$ there is a distribution $\mu$ over $(x,y) \in \{0,1\}^n \times \{0,1\}^n$, such that $D_{\mu, \delta}^\rightarrow(F_0^{\varepsilon}) = \Omega(\frac{\log \delta^{-1}}{\varepsilon^2})$. Moreover, this distribution is supported on vectors with $|x| = |y| = \frac{n}{2}$
\end{theorem}

It is enough now to show that strong tracking algorithm for distinct elements can be leveraged to obtain efficient protocols for $T$-game based on relation $F_0^\varepsilon$.

\begin{lemma}
	\LemmaName{strong-tracking-to-T}
    If there is a randomized streaming algorithm using space $S$ for $(1+\varepsilon)$-strong tracking distinct elements on the universe of size $O(n^2)$, which succeeds with probability $\frac{2}{3}$, then for any distribution $\mu$ supported on pairs $(x, y) \in \{0,1\}^n \times \{0,1\}^n$ of vectors with Hamming weight $\frac{n}{2}$ we have $D_{\mu, \delta}(F_0^{2\varepsilon}) \leq S$.
\end{lemma}
\begin{proof}
		Indeed, consider universe $U$ partitioned into subsets $U_1 \cup U_2 \cup \ldots \cup U_k$, such that $|U_1| = n$, and $|U_i| = 8 |U_{i-1}|$. We can take $k = \Theta(\log n)$ such that $|U| \leq n^2$. Moreover, for each $t \leq k$, consider a partition of $U_t$ into $n$ sets $U_t = U_t^1 \cup \ldots \cup U_t^n$ with $|U_t^i| = 8^{t-1}$. The players are going to pass between each other the memory content of the streaming algorithm. On the $t$-th round, Alice takes her input $x_k$, and feeds to the algorithm all the elements $A_t := \bigcup_{i : (x_k)_i = 1} U_t^{i}$, then she sends the memory content to Bob, who in turn feeds to the algorithm set $B_t := \bigcup_{i : (y_k)_i = 1} U_t^{i}$, and reads off the answer $w$. 

    Let $P_t := \bigcup_{s<t} A_t \cup B_t$, and $D_t = A_t \cup B_t$. Note that Bob knows $|P_t|$, and moreover $|P_t| \leq \sum_{s < t} n 8^s \leq \frac{1}{4} n 8^t \leq \frac{|D_t|}{4}$, where the last inequality follows from the fact that all vectors $x_i$ under consideration have Hamming weight exactly $\frac{n}{2}$.

    By the correctness guarantee of the tracking algorithm, $w$ is a good approximation of $|P_t \cup D_t|$, i.e. $w = (1\pm\varepsilon)(|P_t \cup D_t|)$. Bob can estimate $|D_t|$ by $w - |P_t|$. Indeed: $w - |P_t| \leq (1 + \varepsilon)(|D_t|) + \varepsilon |P_t| \leq (1 + 2 \varepsilon) |D_t|$. Bob can report this estimate to the judge, and send the memory content of the algorithm back to Alice.
\end{proof}

\begin{theorem}
	\TheoremName{lower-bound}
    Any algorithm satisfying $(1+\varepsilon)$ strong tracking of $F_0$ with failure probability at most $\frac{1}{3}$, needs to use at least $\Omega(\frac{\log \log n}{\varepsilon^2})$ bits of space
\end{theorem}
\begin{proof}
	 The statement of this theorem follows directly by composing \LemmaR{strong-tracking-to-T}, \LemmaR{T-to-one-way} and \TheoremR{jayram-woodruff}.
\end{proof}

\section{Pseudorandom construction \SectionName{sampler-lemma}}
In this section we will prove \LemmaR{sampler-lemma}. Before we proceed with the proof, let us introduce a useful definition.

\begin{Definition}
    A function $\Gamma : \{0, 1\}^s \times [w] \to [M]$ is called $\gamma_0$-strong sampler, if for any function $f : [M] \to [0,1]$ and random variables $Y_i := \Gamma(U, i)$ generated by supplying uniformly random $U$, we have for any $2 < \gamma$
    \begin{equation*}
        \P(|\sum_{i \leq w} f(Y_i) > \mu \gamma) \leq \exp(-\Omega(\mu \gamma \log \min\{\gamma, \gamma_0\}))
    \end{equation*}
    and moreover for any fixed $i$, we have $\Gamma(U, i) \sim \Unif([M])$.
\end{Definition}

The definition above is non-vacuous --- as it has been recently shown, standard pseudorandom constructions of samplers actually satisfy our definition of the strong sampler.
\begin{theorem}[\cite{DBLP:journals/corr/RaoR17,DBLP:journals/cpc/Wagner08}]
    \TheoremName{random-walk-sampler}
    A random walk over a finite regular undirected graph with second largest eigenvalue $\lambda$, yields a $\lambda^{-1}$-strong sampler. This implies explicit $\gamma$-strong samplers $\Gamma : \{0, 1\}^s \times [w] \to [M]$ with seed length $s = \log M + \Oh(w \log \gamma)$.
\end{theorem}

In \cite{DBLP:journals/corr/RaoR17} bounds on the moments generating functions of $\sum Y_i$ are proven, instead of the tail bounds that appear in our definition of strong-sampler. They proved
    \begin{equation*}
        \mathbb{E} \exp(\theta \sum f(Y_i)) \leq \exp\left( 2\mu (e^\theta - 1)\right)
    \end{equation*}
    for $\theta \leq \ln \lambda^{-1} - 1$. There is a standard way of deducing tail bounds of the form required for strong samplers from this MGF bound
    \begin{equation*}
        \P(|\sum_{i \leq w} f(Y_i)| > \mu \gamma) \leq \exp(-\mu \gamma \theta) \mathbb{E} \exp(\theta \sum f(Y_i)) \leq \exp(-\mu \gamma \theta + 2 \mu (e^\theta - 1))
    \end{equation*}
    we can plug in $\theta := \ln \gamma^* - \ln 4$, where $\gamma^* := \min\{\gamma, \gamma_0\}$ to get $\P(\sum_{i \leq w} f(Y_i) > \mu \gamma) \leq \exp(- \frac{1}{2} \mu \gamma \ln \gamma^*)$.
    
We will now show that sums of strongly concentrated random variables, sampled according to a strong sampler still satisfy similar type of tail-bounds as if they were sampled independently at random.

\begin{lemma}
    \LemmaName{strong-implies-unbounded-tails}
    If $\Gamma : \{0, 1\}^s \times [w] \to [M]$ is $\exp(\ln^2 w)$-strong sampler, and $f : [M] \to \bR_+$ satisfies doubly exponential tail bounds $\P_{X \sim \mathrm{Unif([M])}}(X > \gamma) \leq \exp(-e^\gamma)$, then
    \begin{equation*}
        \P(\sum_{i \leq w} f(Y_i) > C w) < \exp(- \Omega(w))
    \end{equation*}
    for some universal constant $C$.
\end{lemma}
\begin{proof}
    Take some $T_0$, sufficiently large constant, and consider a sequence $T_k = 2^k T_0$, together with functions $f_k : [M] \to \{0, 1\}$ given by $f_k(x) := [x > T_k]$. Let $\mu_k := w \E f_k(X)$, and notice that because of the assumed tail bounds on function $f$ we have $\mu_k \leq w \exp(-e^{2^k T_0})$. 
    
    We can bound value of $\sum f(Y_i)$ in terms of $f_k$ as follows
    \begin{equation*}
        \sum f(Y_i) \leq \Oh(w) + T_0 \sum_{k\geq 0} 2^k \sum_{i\leq w} f_k(Y_i).
    \end{equation*}

    We shall bound all terms $\sum_i f_k(Y_i)$ separately. Let us take $k_1$ smallest such that $\mu_{k_1} \leq 1$ (i.e. $T_{k_1} \approx \ln \ln w$) and $k_2$ smallest such that $\mu_{k_2} \leq \exp(-w)$ --- we have $T_{k_2} = \Theta(\log w)$.

    Firstly, by Markov inequality $\P(\sum_{i \leq w} f_{k_2}(Y_i) \geq 1) \leq \mu_{k_2} \leq \exp(-w)$, so with probability $1-\exp(-w)$, we have $\sum f(Y_i) \leq \Oh(w) + \sum_{k \leq k_2} 2^k \sum_{i \leq w} f_k(Y_i)$.

    We will bound terms between $0$ and $k_1$, and terms in the range $k_1$ and $k_2$ separately. For $k_1 < k < k_2$, we can use the Chernoff-type inequality guaranteed by the sampler. Indeed, for $k > k_1$, we have $\mu_k < w \exp(-e^{2 \ln \ln w}) < \exp(- \frac{1}{2} (\ln w)^2)$, and therefore if we pick $\gamma_k := \frac{w}{\ln^2 w}\mu_k^{-1} > \exp( (\ln w)^2)$ we have by the definition of strong sampler
    \begin{equation*}
        \P(\sum_{i \leq w} f_k(Y_i) > w/\log^2 w) < \exp(-\Omega(\frac{w}{\log^2 w} \log^2 w)) = \exp(-\Omega(w)).
    \end{equation*}
    If this (exponentially unlikely) event does not hold for any $k$ in this range, we have $\sum_{k_1 < k \leq k_2} \sum_{i\leq w} T_k f_k(Y_i) \leq \sum_{k_1 \leq k \leq k_2} T_0 \ln w \frac{w}{\ln^2 w} = \Oh(w)$, because $k_2 = \Oh(\log w)$.

    Let us now focus on the range $k < k_1$, and let us consider $\gamma_k$ such that $\gamma_k \mu_k = 3^{-k} w$. For $k < k_1$ we have $\gamma_k < w < \exp(\ln^2 w)$, so the sampler guarantee gives us
    \begin{equation}
        \EquationName{last-failure}
        \P(\sum_{i \leq w} f(Y_i) > \gamma_k \mu_k) \leq \exp(- 3^{-k} w \log \gamma_k)
    \end{equation}

    Clearly, if neither of those events hold, we have
    \begin{equation*}
        \sum_{k \leq k_1} \sum_{i \leq w} T_k f(Y_i) \leq \sum_{k \leq k_1} T_0 2^k 3^{-k} w = \Oh(w)
    \end{equation*}

    It is enough to bound the failure probability in \Equation{last-failure}. We have $\log \gamma_k = -k \log 3 + \log w - \log \mu_k = - k \log 3 + \log \P(f(X) > T_k) > k \log 3$. As such, for any fixed $k$, we have $\P(\sum_i \leq w f(Y_i) > \gamma_k \mu_k \leq \exp(-w)$, and by union bound the failure probability is bounded by $\exp(-w + \log k_1) < \exp(-\Omega(w))$.
\end{proof}

First, let us observe that \TheoremR{random-walk-sampler} and \LemmaR{strong-implies-unbounded-tails} implies that
\begin{lemma}
    \LemmaName{random-walk-small}
    For $w_2 \geq K \log R$, there exist an explicit function $\Gsamp_0 : \{0, 1\}^{s_0} \times [w_2] \to [M]$ such that set
    $\{\Gsamp_0(U,1), \ldots \Gsamp_0(U, w_2)\}$ is $C$-small except with probability $\exp(-c w_2)$.

    The seed length is $s_0 = \Oh(\log M + w_2 \log^2 w_2)$.
\end{lemma}
\begin{proof}
    Consider $\Gsamp_0$ as in \TheoremR{random-walk-sampler} with parameter $\lambda = \Oh(\exp(\log^2 w_2))$. We know that for every specific $g_t$, with probability $\exp(-\Omega(w_2))$, the sum over the generated sequence satisfies $\sum_{t \leq w_2} g_t(\Gsamp_0(U, i)) \leq C w_2$. We can union bound over all $t \leq R$, so the probability that $\Gsamp_0(U, *)$ fails to be $C$-small is bounded by $\exp(-\Omega(w_2) + \log R) = \exp(-\Omega(w_2))$, as long as $K$ in the statement of the lemma is sufficiently large constant.
\end{proof}

In what follows we will use as a building block the construction guaranteed by the following theorem
\begin{theorem}[\cite{DBLP:journals/rsa/Zuckerman97, DBLP:journals/jacm/GuruswamiUV09} \cite{Vadhan12} Corollary 6.24]
    There exist an explicit $(\varepsilon, \delta)$-averaging sampler $\Gamma : \{0, 1\}^s \times [W] \to [M]$, with $s = \Oh(\log M + \log \delta^{-1})$ and $W = \poly(\varepsilon^{-1}, \log \delta^{-1}, \log M)$.
\end{theorem}
We shall use such a sampler to subsample a set of seeds for the expander random walks discussed in \LemmaR{random-walk-small}. We can ensure that except with probability $\exp(-\Omega(w))$ the subsampled set of seeds has the same fraction of seeds generating $C$-small sets.

\begin{lemma}
    \LemmaName{reduced-universe}
	For any $w > K R$ and $w_2$ satisfying $K \log R < w_2 < w$, there exist $c$ and an explicit function $\Gsamp_1 : \{0, 1\}^{s_1} \times [W] \to [M]^{w_2}$ such that
    \begin{equation*}
			\P_{x \in \Unif(\{0,1\}^{s_1})}( \#\{ w : G_1(x, w) \text{ is not $C$-small}\} > 2 \exp(-c w_2) W) < \exp(-\Omega(w))
    \end{equation*}
	The seed length here is $s_1 := \Oh(\log M + w)$ and $W = \poly(w, \exp(w_2), \log M)$.
\end{lemma}
\begin{proof}
    Take $\Gamma$ to be $(\exp(-c w_2), \exp(-w))$-averaging sampler, where $c$ is such that $\Gsamp_0$ from \LemmaR{random-walk-small} provides a $C$-small set except with probability $\exp(-c w_2)$. Consider function $\Gsamp_1(S, i) = \Gsamp_0(\Gamma(S, i), *)$, i.e. $[\Gsamp_1(S, i)]_j = \Gsamp_0(\Gamma(S, i), j)$. The required properties follow from definition of the averaging sampler applied to the indicator function of $f :\{0, 1\}^{s_0} \to \{0,1\}$ with $f(s) = 1$ if and only if $\Gsamp_0(S, *)$ yields a $C$-small sequence. 
\end{proof}

Finally, we are ready to prove the main lemma in this section.

\begin{proof}[Proof of \LemmaR{sampler-lemma}]
		Given $w > K R^{1/2}$, take $w_2 = \Theta(\log w)$ large enough to apply \LemmaR{reduced-universe}, and $w_1$ large enough with $w_1 w_2 = \Theta(w)$. Take $\Gsamp_1$ as in the \LemmaR{reduced-universe}, and note that with the setting of parameters $w_2 = \Theta(\log w)$ and $w \geq (\log M)^{\Omega(1)}$, we have in fact $W = w^{O(1)}$. Consider the decomposition of the seed $S \in \{0,1\}^s$ as $S = (S_1, S_2) \in \{0,1\}^s = \{0,1\}^{s_1 + s_2}$, and let us focus on collection $\mathcal{A} = \Gsamp_1(S_1, *)$ of $W$ sequences in $[M]$. We know that, except with probability $\exp(-\Omega(w))$ over choice of the seed $S_1$, we most of the sequences in $\mathcal{A}$ is $C$-small --- only $2 \exp(-c w_2)$ fraction of all sequences is not $C$-small. Let us use $S_2$ to pick a uniformly random sequence of indices from $[W]$ --- to achieve this, we need $s_2 = \Theta(w_1 \log W)$. Note that if $\mathcal{A}$ indeed satisfies that $\#\{i : [\mathcal{A}]_i \text{ is $C$-small} < 2 \exp(-c w_2)\}$, then for a uniformly random indices $i_1, i_2, \ldots i_{w_1} \in [W]$, we have

    \begin{equation*}
        \P(\#\{ j : [\mathcal{A}]_{i_j}] \text{ is $C$-small}\} > w_1/2) \leq \binom{w_1}{w_1/2} \exp(-w_2)^{w_1/2} \leq 2^{w_1} \exp(-\Omega(w)) = \exp(-\Omega(w))
    \end{equation*}

    The total seed length is $s = s_1 + s_2 = \Oh(\log M + w) + \Oh( w_1 \log W) = \Oh(\log M + w + \frac{w}{\log w} \log w) = \Oh(\log M + w)$.
\end{proof}

\textbf{Acknowledgement}
The author thanks Raghu Meka for answering questions about the \cite{DBLP:conf/soda/Meka17} sampler construction, which inspired the proof of \LemmaR{sampler-lemma}, Preetum Nakkiran for helpful discussion during various stages of the work, Thibaut Horel for comments on parts of the writeup, Rohit Agrawal for additional discussions about the sampler construction.
The author is especially grateful to Jelani Nelson for many inspiring and helpful discussions and comments.

\bibliographystyle{alpha}
\bibliography{biblio}

\newcommand{\etalchar}[1]{$^{#1}$}
\begin{thebibliography}{BCWY16}

\bibitem[AMS96]{DBLP:conf/stoc/AlonMS96}
Noga Alon, Yossi Matias, and Mario Szegedy.
\newblock The space complexity of approximating the frequency moments.
\newblock In Gary~L. Miller, editor, {\em Proceedings of the Twenty-Eighth
  Annual {ACM} Symposium on the Theory of Computing, Philadelphia,
  Pennsylvania, USA, May 22-24, 1996}, pages 20--29. {ACM}, 1996.

\bibitem[BC09]{DBLP:journals/eccc/BrodyC09}
Joshua Brody and Amit Chakrabarti.
\newblock A multi-round communication lower bound for gap hamming and some
  consequences.
\newblock {\em Electronic Colloquium on Computational Complexity {(ECCC)}},
  16:15, 2009.

\bibitem[BCI{\etalchar{+}}17]{BravermanCINWW17}
Vladimir Braverman, Stephen~R. Chestnut, Nikita Ivkin, Jelani Nelson, Zhengyu
  Wang, and David~P. Woodruff.
\newblock {BPTree}: an $\ell_2$ heavy hitters algorithm using constant memory.
\newblock In {\em Proceedings of the 36$^{th}$ SIGMOD-SIGACT-SIGART Symposium
  on Principles of Database Systems (PODS)}, 2017.

\bibitem[BCWY16]{BravermanCWY16}
Vladimir Braverman, Stephen~R. Chestnut, David~P. Woodruff, and Lin~F. Yang.
\newblock Streaming space complexity of nearly all functions of one variable on
  frequency vectors.
\newblock In {\em Proceedings of the 35$^{th}$ {ACM} {SIGMOD-SIGACT-SIGAI}
  Symposium on Principles of Database Systems (PODS)}, pages 261--276, 2016.

\bibitem[BDN17]{DBLP:journals/corr/DingBN17}
Jaroslaw Blasiok, Jian Ding, and Jelani Nelson.
\newblock Continuous monitoring of $\ell_p$ norms in data streams.
\newblock {\em CoRR}, abs/1704.06710, 2017.

\bibitem[BJK{\etalchar{+}}02]{DBLP:conf/random/Bar-YossefJKST02}
Ziv Bar{-}Yossef, T.~S. Jayram, Ravi Kumar, D.~Sivakumar, and Luca Trevisan.
\newblock Counting distinct elements in a data stream.
\newblock In Jos{\'{e}} D.~P. Rolim and Salil~P. Vadhan, editors, {\em
  Randomization and Approximation Techniques, 6th International Workshop,
  {RANDOM} 2002, Cambridge, MA, USA, September 13-15, 2002, Proceedings},
  volume 2483 of {\em Lecture Notes in Computer Science}, pages 1--10.
  Springer, 2002.

\bibitem[BR94]{BellareR94}
Mihir Bellare and John Rompel.
\newblock Randomness-efficient oblivious sampling.
\newblock In {\em Proceedings of the 35$^{th}$ Annual {IEEE} Symposium on
  Foundations of Computer Science (FOCS)}, pages 276--287, 1994.

\bibitem[DF03]{Durand2003}
Marianne Durand and Philippe Flajolet.
\newblock {\em Loglog Counting of Large Cardinalities}, pages 605--617.
\newblock Springer Berlin Heidelberg, Berlin, Heidelberg, 2003.

\bibitem[EVF06]{DBLP:journals/ton/EstanVF06}
Cristian Estan, George Varghese, and Michael~E. Fisk.
\newblock Bitmap algorithms for counting active flows on high-speed links.
\newblock {\em {IEEE/ACM} Trans. Netw.}, 14(5):925--937, 2006.

\bibitem[FFGM07]{flajolet2007hyperloglog}
Philippe Flajolet, {\'E}ric Fusy, Olivier Gandouet, and Fr{\'e}d{\'e}ric
  Meunier.
\newblock Hyperloglog: the analysis of a near-optimal cardinality estimation
  algorithm.
\newblock In {\em AofA: Analysis of Algorithms}, pages 137--156. Discrete
  Mathematics and Theoretical Computer Science, 2007.

\bibitem[FM83]{DBLP:conf/focs/FlajoletM83}
Philippe Flajolet and G.~Nigel Martin.
\newblock Probabilistic counting.
\newblock In {\em 24th Annual Symposium on Foundations of Computer Science,
  Tucson, Arizona, USA, 7-9 November 1983}, pages 76--82. {IEEE} Computer
  Society, 1983.

\bibitem[Gar07]{Garling}
D.~J.~H. Garling.
\newblock {\em Inequalities: A Journey into Linear Analysis}.
\newblock Cambridge University Press, 2007.

\bibitem[GG81]{gabber1981explicit}
Ofer Gabber and Zvi Galil.
\newblock Explicit constructions of linear-sized superconcentrators.
\newblock {\em Journal of Computer and System Sciences}, 22(3):407--420, 1981.

\bibitem[Gib01]{DBLP:conf/vldb/Gibbons01}
Phillip~B. Gibbons.
\newblock Distinct sampling for highly-accurate answers to distinct values
  queries and event reports.
\newblock In {\em {VLDB} 2001, Proceedings of 27th International Conference on
  Very Large Data Bases, September 11-14, 2001, Roma, Italy}, pages 541--550.
  Morgan Kaufmann, 2001.

\bibitem[Gil98]{Gillman:1998:CBR:284943.284979}
David Gillman.
\newblock A chernoff bound for random walks on expander graphs.
\newblock {\em SIAM J. Comput.}, 27(4):1203--1220, August 1998.

\bibitem[GT01]{DBLP:conf/spaa/GibbonsT01}
Phillip~B. Gibbons and Srikanta Tirthapura.
\newblock Estimating simple functions on the union of data streams.
\newblock In {\em {SPAA}}, pages 281--291, 2001.

\bibitem[GUV09]{DBLP:journals/jacm/GuruswamiUV09}
Venkatesan Guruswami, Christopher Umans, and Salil~P. Vadhan.
\newblock Unbalanced expanders and randomness extractors from parvaresh-vardy
  codes.
\newblock {\em J. {ACM}}, 56(4):20:1--20:34, 2009.

\bibitem[HTY14]{HuangTY14}
Zengfeng Huang, Wai~Ming Tai, and Ke~Yi.
\newblock Tracking the frequency moments at all times.
\newblock {\em CoRR}, abs/1412.1763, 2014.

\bibitem[JW13]{DBLP:journals/talg/JayramW13}
T.~S. Jayram and David~P. Woodruff.
\newblock Optimal bounds for johnson-lindenstrauss transforms and streaming
  problems with subconstant error.
\newblock {\em {ACM} Trans. Algorithms}, 9(3):26:1--26:17, 2013.

\bibitem[KNW10]{DBLP:conf/pods/KaneNW10}
Daniel~M. Kane, Jelani Nelson, and David~P. Woodruff.
\newblock An optimal algorithm for the distinct elements problem.
\newblock In Jan Paredaens and Dirk~Van Gucht, editors, {\em Proceedings of the
  Twenty-Ninth {ACM} {SIGMOD-SIGACT-SIGART} Symposium on Principles of Database
  Systems, {PODS} 2010, June 6-11, 2010, Indianapolis, Indiana, {USA}}, pages
  41--52. {ACM}, 2010.

\bibitem[Mek17]{DBLP:conf/soda/Meka17}
Raghu Meka.
\newblock Explicit resilient functions matching ajtai-linial.
\newblock In Philip~N. Klein, editor, {\em Proceedings of the Twenty-Eighth
  Annual {ACM-SIAM} Symposium on Discrete Algorithms, {SODA} 2017, Barcelona,
  Spain, Hotel Porta Fira, January 16-19}, pages 1132--1148. {SIAM}, 2017.

\bibitem[RR17]{DBLP:journals/corr/RaoR17}
Shravas Rao and Oded Regev.
\newblock A sharp tail bound for the expander random sampler.
\newblock {\em CoRR}, abs/1703.10205, 2017.

\bibitem[Vad12]{Vadhan12}
Salil~P. Vadhan.
\newblock Pseudorandomness.
\newblock {\em Foundations and Trends in Theoretical Computer Science},
  7(1-3):1--336, 2012.

\bibitem[Wag08]{DBLP:journals/cpc/Wagner08}
R.~O.~Y. Wagner.
\newblock Tail estimates for sums of variables sampled by a random walk.
\newblock {\em Combinatorics, Probability {\&} Computing}, 17(2):307--316,
  2008.

\bibitem[Woo04]{Woodruff:2004:OSL:982792.982817}
David Woodruff.
\newblock Optimal space lower bounds for all frequency moments.
\newblock In {\em Proceedings of the Fifteenth Annual ACM-SIAM Symposium on
  Discrete Algorithms}, SODA '04, pages 167--175, Philadelphia, PA, USA, 2004.
  Society for Industrial and Applied Mathematics.

\bibitem[Zuc97]{DBLP:journals/rsa/Zuckerman97}
David Zuckerman.
\newblock Randomness-optimal oblivious sampling.
\newblock {\em Random Struct. Algorithms}, 11(4):345--367, 1997.

\end{thebibliography}

\newpage
\appendix

\section{Appendix \label{sec:appendix}}

In this section, for the sake of completeness, we will discuss space complexity of the optimal algorithm with constant probability proposed in \cite{DBLP:conf/pods/KaneNW10} --- we use it as a building block in \Section{high-accuracy}. The existence of the algorithm as described below was proven in \cite{DBLP:conf/pods/KaneNW10}, as well as the fact that it returns correct answer with large constant probability. In what follows we describe the KNW algorithm, and provide more detailed analysis of the space complexity of this algorithm --- in the original paper, it was shown only that for $\frac{1}{\varepsilon^2} > \log n$, the total space consumption is $\Oh(\frac{1}{\varepsilon^2})$  with large constant probability. The condition $\frac{1}{\varepsilon^2} > \log n$ could have been assumed without loss of generality in the original setting, it is not the case in our application.

The correctness of this algorithm (with large constant probability) has been shown in \cite{DBLP:conf/pods/KaneNW10}, it also follows from the proofs in \Section{strong-tracking}. In particular in the proof of \TheoremR{strong-tracking} it is shown how to deduce strictly stronger statement. We do not discuss it in this appendix.

\knwreduction*
\begin{proof}
	Consider some $P = \frac{C_0}{\varepsilon^2}$ with large constant $C_0$ (depending on $C$) and some constant $D_0$ to be specified later. We will pick a random pairwise independent hash function $h_3 : [n] \to [P^2]$, and random $\polylog(P)$-wise independent hash function $h_4 : [P^2] \to [P]$. We set $h_2 := h_4 \circ h_3$, and we take $h_1 : [n] \to [n]$ to be 8-wise independent hash function. The total number of random bits necessary to access is $\Oh(\log n + \poly \log P) = \Oh(\log n + \poly\log\frac{1}{\varepsilon})$. 
	

We assume access to $\tilde{F}_0^{(t)}$ such that for every $t$ we have $F_0^t \leq \tilde{F}_0^{(t)} \leq C F_0^{(t)}$. For each $i \in [P]$ we consider $Z_i^{(t)} := \max \{\lsb(h_1(s)) : s \in S^{(t)}, h_2(s) = i\}$, and we store $\hat{Z}_i^{(t)} := \max \{-1, Z_i^{(t)} - D^{(t)}\}$, where $D^{(t)} := \log \tilde{F}_0^{(t)} - \log \frac{1}{\varepsilon^2} - D_0$. At time $t$, we consider $Q^{(t)} := \#\{ i : Z_i^{(t)} \geq 0\}$, and we report $\hat{F}^{(t)} := \Phi^{-1}_{K}(Q^{(t)}) 2^{D^{(t)}}$.


	Let us consider total space used by all the counters $Z_i$. We need to use $\Oh(P)$ bits to store all counters for which value of $Z_i^{(t)}$ is $-1$, and space necessary to store counters with $Z_i^{(t)} \geq 0$ is bounded by $\sum_{i \in S^{(t)}} \log \max(1, \lsb(h_1(s)) - D^{(t)})$. Hence, the space used by the algorithm at time $t$ is bounded as
\begin{equation*}
W^{(t)} \leq \Oh(P) + \sum_{i \in S^{(t)}} \log \max(1, \lsb(h_1(s)) - D^{(t)}) \leq \Oh(P) + \sum_{i \in S^{(t)}} \max(0, \lsb(h_1(s)) - D^{(t)})
\end{equation*}
Note that for fixed $t$, variables $K_s := \max(0, \lsb(h_1(s)) - D^{(t)})$ are $8$-wise independent (because $h_1$ was). For fixed $s$ the random variable $K_s$ have strongly decaying tails, i.e. $\P(K_s > \lambda) = \P(\lsb(h_1(s)) > D^{(t)} + \lambda) \leq 2^{- D^{(t)} - \lambda} \leq \frac{2^{D_0}}{\varepsilon^2 \hat{F}^{(t)}} 2^{-\lambda}$. We can apply \LemmaR{supber-concentration} to appropriately rescaled $K_s$, with $p=8$ and $\mu_i = \frac{2^{D_0}}{\varepsilon^2 \hat{F}^{(t)}}$. With this choice of $\mu_i$ we have $\frac{1}{\varepsilon^2} \cdot \frac{2^{D_0}}{C} \leq \sum \mu_i \leq \frac{1}{\varepsilon^2} 2^{D_0}$, so the conditions of this lemma are satisfied, and the conclusion of the lemma yields
\begin{equation*}
	\P(\sum_{s \in S^{(t)}} K_s > C_{1} \frac{1}{\varepsilon^2}) \leq C_2 \varepsilon^4 2^{-D_0}
\end{equation*}
We can now take $D_0$ such that that $2^{-D_0} C_2 \leq 1$ to finish the proof of~\Equation{small-eps}.

Let us now turn our attention to the proof of~(\ref{eq:large-eps}). Observe that $\P(W^{(t)} > \frac{\lambda}{\varepsilon^2}) \leq \P(\exists i\in [P], s.t. \log \hat{Z}_i > C_1 \lambda) = \P(\exists s \in S^{(t)}, s.t. \lsb(h_1(s)) > D^{(t)} + 2^{C_1 \lambda})$, where $C_1$ is a constant depending on $C_0$. Using union bound, this latter quantity is bounded as follows
\begin{equation*}
	\P(\exists s \in S^{(t)}, s.t. \lsb(h_1(s)) > D^{(t)} + 2^{C \lambda}) \leq F_0^{(t)} \P(\lsb(h_1(s_0)) > P^{(t)} + 2^{C_1 \lambda}) \leq \frac{2^{D_0} C}{\varepsilon^2} 2^{- 2^{C_1 \lambda}}.
\end{equation*}

Which yields a bounds of form $\P(W^{(t)} \geq \frac{\lambda}{\varepsilon^2}) \leq \frac{C_2}{\varepsilon^2} 2^{-2^{C_1 \lambda}}$. On the other hand, by \LemmaR{supber-concentration} we have $\P(W^{(t)} \geq \frac{\lambda}{\varepsilon^2}) \leq \frac{\varepsilon^4}{\lambda^8}$. We can combine those two bounds for different ranges of $\lambda$ to get $\P(W^{(t)} \geq \frac{\lambda}{\varepsilon^2} ) \leq \exp(-e^{\Omega(\lambda)})$. Indeed, for $\lambda < \log \log \frac{1}{\varepsilon^2}$ we already have $\frac{1}{\varepsilon^4} < \exp(-e^{\Omega(\lambda)})$ whereas for $\lambda > \log \log \frac{1}{\varepsilon^2}$ we have $\frac{1}{\varepsilon^2} \exp(-e^{\lambda}) < \exp(-e^{\lambda/2})$.

Finally, to show~(\ref{eq:timestops}), note that $W^{(t_1)} \leq \sum_{i} \lceil \log (Z_{i}^{(t_1)} - D^{(t_1)}) \rceil \leq \sum_i \lceil \log (Z_{i}^{(t_2)} - D^{(t_1)}$, because $Z_{i}^{(t_2)} \geq Z_{i}^{(t_1)}$. By subadditivity of logarithm, we have $W^{(t_1)} \leq \sum_i \lceil \log (Z_{i}^{(t_2)} - D^{(t_2)} \rceil + P (D^{(t_2)} - D^{(t_1)}) \leq W^{(t_2)} + \Oh(P)$, where $(D^{(t_2)} - D^{(t_1)} = \Oh(1)$ follows from the fact that $D^{(t_2)} - D^{(t_1)} = \log \frac{\hat{F}_0^{(t_2)}}{\hat{F}_0^{(t_1)}}$, and since $\hat{F}^{(t)}$ is a constant approximation to $F_0^{(t)}$, this quantity is bounded by $\log \frac{F_0^{(t_2)}}{F_0^{(t_1)}} + \Oh(1)$, and by assumption on $t_2, t_1$, we have $\log \frac{F_0^{(t_2)}}{F_0^{(t_1)}} \leq \Oh(1)$.
\end{proof}

\subsection{Probabilistic inequalities}
\begin{lemma}
	Let $Z_1, \ldots Z_k$ be a sequence of non-negative, $p$-wise independent random variables (for some even $p$), satisfying $\|Z\|_p \leq C$. Then
	\begin{equation*}
		\|\sum (Z_i - \E Z_i)\|_p \lesssim C\sqrt{p} \sqrt{k}
	\end{equation*}
	\label{lem:p-wise-indep-sum}
\end{lemma}
\begin{proof}
	Take $Y_i$ independent, with marginal distribution $Y_i \sim Z_i - \E Z_i$. Because $\|\sum (Z_i - \E Z_i)\|_p^p$ is a polynomial of degree $p$ in variables $Z_i$, and $Y_i$ are independent, it follows that $\|\sum (Z_i - \E Z_i)\|_p = \|\sum Y_i\|_p$, and it is enough to bound this second quantity. We can use symmetrization argument, to deduce that $\|\sum Y_i\|_p \lesssim \|\sum \varepsilon_i Y_i\|_p$  where $\varepsilon_i$ are independent random signs. 
	
	Indeed, consider $\tilde{Y}_i$ distributed identically as $Y_i$ and independent from those, then 
	\begin{align*}
		\|\sum Y_i\|_p & = \| \sum (Y_i - \E \tilde{Y}_i)\|_p 
		 \leq \|\sum (Y_i - \tilde{Y}_i)\|_p \\
		& = \| \sum \varepsilon_i (Y_i - \tilde{Y}_i)\|_p 
		 \leq 2 \|\sum \varepsilon_i Y_i\|_p.
	\end{align*}

	We can now condition on $Y_i$ and use Khnitchine inequality \cite[Theorem 12.3.1]{Garling} to bound
	\begin{align*}
		\|\sum \varepsilon_i Y_i\|_p & = \left( \E (\sum \varepsilon_i Y_i)^p \right)^{1/p} \\
		& \lesssim \sqrt{p} \left( \E(\sum Y_i^2)^{p/2} \right)^{1/p} \\
		& = \sqrt{p} \sqrt{\| \sum Y_i^2 \|_{p/2}} \\
		& \leq \sqrt{p} \sqrt{\sum \|Y_i^2\|_{p/2}} \\
		& \leq \sqrt{p} \sqrt{ \sum \|Y_i\|_p^{2}} \\
		& \leq C \sqrt{p} \sqrt{k}
	\end{align*}
\end{proof}

\begin{lemma}
	\label{lem:supber-moment-bound}
	For every $p$ there exist $C_p, \tilde{C}_p$ such that if non-negative independent random variables $Z_1, \ldots Z_k$ satisfy $\E Z_i^{s} \leq \mu_i$ for all $1\leq s \leq p$, where $\sum \mu_i \geq 1$ then
	\begin{equation}
		\|\sum Z_i\|_p \leq C_p \sum \mu_i \label{eq:small-p-norm}
	\end{equation}
	and moreover
	\begin{equation}
		\|\sum (Z_i - \E Z_i)\|_p \leq \tilde{C}_p \sqrt{\sum \mu_i} \label{eq:sqrt-deviation}
	\end{equation}
\end{lemma}
\begin{proof}
	It is enough to prove ineqialities~(\ref{eq:small-p-norm}) and~(\ref{eq:sqrt-deviation}) for all values $p$ that are powers of two. We will proceed by showing~(\ref{eq:small-p-norm}) by induction over $p$. The case $p=1$ is trivial: $\|\sum Z_i\|_1 = \sum \E Z_i \leq \sum \mu_i$. For $p>1$, let us take $Y_i := Z_i - \E Z_i$. We have
	\begin{equation}
		\|\sum Z_i\|_p \leq \sum \E Z_i + \|\sum Y_i\|_p \label{eq:pull-mean}
	\end{equation}
	Now we can use standard symmetrization argument to bound $\|\sum Y_i\|_p$. Let us take $\tilde{Y}_i$ to be independent random variables with the same distribution as $Y_i$, and $\varepsilon_i$ to be independent uniform $\pm 1$ random variables. We have
	\begin{align}
		\|\sum Y_i\|_p &= \|\sum (Y_i - \E \tilde{Y}_i)\|_p \leq \|\sum (Y_i - \tilde{Y}_i)\|_p \nonumber \\
		& = \|\sum \varepsilon_i (Y_i - \tilde{Y}_i)\|_p \leq 2 \|\sum \varepsilon_i Y_i\|_p \label{eq:symmetrization}
	\end{align}
	We can now condition on $Y_i$ and use Khintchine inequality to deduce
	\begin{align*}
		\|\sum \varepsilon Y_i\|_p = \left( \E (\sum \varepsilon_i Y_i)^p\right)^{1/p}  
		\lesssim \sqrt{p} \| \sum Y_i^2\|_{p/2}^{1/2} \leq \sqrt{p} \|\sum Z_i^2 \|_{p/2}^{1/2}
	\end{align*}
	By applying inductive hypothesis to random variables $Z_i^2$ we obtain
	\begin{equation*}
		\|\sum Y_i\|_p \leq 2 \|\sum \varepsilon Y_i\|_p \lesssim \sqrt{p} \sqrt{C_{p/2}} \sqrt{\sum \mu_i}
	\end{equation*}
	proving inequality~(\ref{eq:sqrt-deviation}). Finally, we can can compose this last inequality with inequality~(\ref{eq:pull-mean}), to deduce
	\begin{equation*}
		\|\sum Z_i\|_p \leq \sum \mu_i + K \sqrt{p} \sqrt{C_{p/2}} \sqrt{\sum \mu_i} \leq (1 + K \sqrt{p C_{p/2}})(\sum \mu_i)
	\end{equation*}
	which completes the proof of inductive hypothesis with $C_{p} = (1 + K \sqrt{p C_{p/2}})$.
\end{proof}

\begin{lemma}
	\label{lem:moments-of-subexp}
	Let $Z$ be a non-negative random variable satisfying for some $T$, that $\P(Z > \lambda T) \leq \mu \exp(-\lambda)$. Then $\E Z^p \leq p! T^p \mu$.
\end{lemma}
\begin{proof}
	We can assume without loss of generality that $T=1$. We can bound
	\begin{align*}
		\E Z^p & = \int_0^{\infty} t^{p-1} \P(Z > t) \,\d t\\
		& \leq \mu \int_0^{\infty} t^{p-1} e^{-t} \, \d t.
	\end{align*}
	Now by repeatedly applying integration by parts, we obtain
	\begin{align*}
		\int_0^{\infty} t^{p-1} e^{-t} \, \d t
		& = (p-1) \int_0^{\infty} t^{p-2} e^{-t} \, \d t \\
		& = \cdots \\
		& = (p-1)! \int_0^{\infty} e^{-t}\, \d t = (p-1)!,
	\end{align*}
	which completes the proof of the desired inequality.
\end{proof}
\begin{corollary}
	\label{cor:supber-moment-bound}
	Let $Z_1, \ldots Z_k$ be a sequence of non-negative random variables satisftying for some $T$ that $\P(Z_i > T \lambda) \leq \mu_i \exp(-\lambda)$. Then $\|\sum (Z_i - \E Z_i)\|_p \lesssim \tilde{C}_p T \sqrt{\sum \mu_i}$, where $\tilde{C}_p$ is a constant that depends only on $p$.
\end{corollary}
\begin{proof}
	Follows directly from Lemma~\ref{lem:moments-of-subexp} and Lemma~\ref{lem:supber-moment-bound}.
\end{proof}

\begin{lemma}
	\LemmaName{supber-concentration}
	Let $Z_1, \ldots Z_k$ be a sequence of $p$-wise independent non-negative random variables, satisfying $\P(Z_i > \lambda) \leq \mu_i \exp(-\lambda)$. Then for some universal constant $K$, and constant $\tilde{C}_p$ depending only on $p$ we have for all $\lambda > K$ following tail bound
	\begin{equation*}
		\P(\sum Z_i > \lambda \sum \mu_i) \leq \frac{1}{\lambda^p} \left(\frac{\tilde{C}_p}{\sum \mu_i}\right)^{p/2}.
	\end{equation*}
\end{lemma}
\begin{proof}
	Note that for random variables as above we have $\E Z_i \leq K_0 \mu_i$ for some universal constant $\mu_i$. Let us pick $K = 2 K_0 $, such that $\P(\sum Z_i > \lambda \sum \mu_i) \leq \P(\sum (Z_i - \E Z_i) > \frac{\lambda}{2} \sum \mu_i)$. By Chebyshev inequality we have
	\begin{equation*}
		\P(\sum (Z_i - \E Z_i) > \frac{\lambda}{2} \sum \mu_i) \leq \left(2 \frac{\|\sum (Z_i - \E Z_i)\|_p}{\lambda \sum \mu_i}\right)^{p}.
	\end{equation*}
	
	We can bound the numerator in this expression using Corollary~\ref{cor:supber-moment-bound}, i.e. $\|\sum (Z_i - \E Z_i)\|_p \leq  \tilde{C}_p \sqrt{\sum \mu_i}$, to deduce desired probability bound.
\end{proof}

\begin{lemma}
		\label{lem:azuma}
		Let $f : \Sigma^n \to \bR$ be function with bounded differences (i.e. for any $\sigma_1, \ldots \sigma_n$ and $\sigma'_j$ we have $|f(\sigma_1, \ldots \sigma_n) - f(\sigma_1, \ldots, \sigma'_j, \ldots \sigma_n)| \leq 1$), and let $X_1, \ldots X_n$ be a sequence of independent random variables.  Consider a Doob martingale $S_j := \E [ f(X_1, \ldots X_n) ~|~ X_1, \ldots X_j]$. Then for any $i, j \leq n$ we have
		\begin{equation*}
			\|S_i - S_j\|_p \lesssim \sqrt{p} \sqrt{i - j}
		\end{equation*}
\end{lemma}
\begin{proof}
		Note that martingale $S_j$ has bounded increments: $S_j - S_{j-1} \leq 1$ with probability 1. By Azuma inequality random variable $S_i - S_j$ is subgaussian with variance $\Oh(|i-j|)$, and hence the moments are bounded by those of corresponding gaussian.
\end{proof}

\end{document}